\pdfoutput=1

\documentclass[a4paper]{article}

\usepackage{proceed2e}
\usepackage[margin=1in]{geometry}
\usepackage{colonequals}
\usepackage{newtxtext,newtxmath}
\usepackage{cmap}
\usepackage{amsmath,amsfonts,amssymb}
\usepackage{graphicx}
\usepackage{microtype}
\usepackage{xspace}
\usepackage{paralist}
\usepackage{enumitem}
\usepackage[ruled,vlined]{algorithm2e}
\usepackage{wrapfig}
\usepackage[thmmarks]{ntheorem}
\usepackage{subfigure}
\usepackage{booktabs}
\usepackage{nicefrac}
\usepackage{tikz}
\usepackage{multirow}
\usepackage{algorithm2e,theorem,balance}
\usepackage{caption}

\usetikzlibrary{arrows,decorations.pathmorphing,positioning,fit,trees,shapes,shadows,automata,calc}

\DontPrintSemicolon
\SetKwInOut{Input}{Input}\SetKwInOut{Output}{Output}

\SetCommentSty{mycommfont}

\tikzset{>=stealth}
\tikzset{state/.append style={inner sep=2pt,minimum size=2pt}}

\newtheorem{definition}{Definition}
\newtheorem{lemma}{Lemma}
\newtheorem{theorem}{Theorem}
\newtheorem{example}{Example}
\newtheorem{corollary}{Corollary}

\newtheorem{proposition}{Proposition}

\renewenvironment{thebibliography}[1]{
  \begin{oldthebibliography}{#1}
    \setlength{\itemsep}{0em}
    \setlength{\parskip}{0em}
}
{
  \end{oldthebibliography}
}

\usepackage[
    pdfauthor={S. Junges, N. Jansen, R. Wimmer, T. Quatmann, L. Winterer, J.-P. Katoen, B. Becker},
    pdftitle={Permissive Finite-State Controllers of POMDPs using Parameter Synthesis},
    pdflang={en-US},
           ]{hyperref}

\usepackage{array}
\newcolumntype{H}{>{\setbox0=\hbox\bgroup}c<{\egroup}@{}}

\usepackage{etoolbox}

\newtoggle{TR}
\toggletrue{TR}

\renewcommand{\paragraph}[1]{\par\smallskip\noindent\textbf{#1}}

\newtheorem{problem}{Problem}

\newcommand{\ie}{i.\,e.\@\xspace}

\newcommand{\eg}{e.\,g.\@\xspace}

\newcommand{\prophesy}{\textrm{PROPhESY}\xspace}

\newcommand{\storm}{\textrm{Storm}\xspace}

\newcommand{\tool}[1]{\textrm{#1}\xspace}

\newcommand{\TO}{TO}
\newcommand{\MO}{MO}
\newcommand{\Err}{Err}
\newcommand{\polred}{\ensuremath{\leqslant_P}}

\newcommand{\remaining}{\ensuremath{\mathit{Remain}}}
\newcommand{\MC}{{D}}
\newcommand{\dtmc}{\MC}

\newcommand{\pDtmcInit}[1][]{\ensuremath{\dtmc{#1}=(S{#1},\sinit{#1},\Paramvar,P{#1})}}

\newcommand{\sinitMk}[1][\pomdp,k]{s_{\mathrm{I},#1}}
\newcommand{\VarMk}[1][\pomdp,k]{V_{#1}}
\newcommand{\pDtmcInitMk}[1][\pomdp,k]{\ensuremath{\dtmc_{#1} = (S_{#1},\sinitMk[#1],\VarMk[#1],P_{#1})}}

\newcommand{\VarMkSubs}[1][\pomdp,k]{V_{#1}^{\mathrm{subs}}}
\newcommand{\VarMkNext}[1][\pomdp,k]{V_{#1}^{\mathrm{next}}}
\newcommand{\VarMkRestr}[1][\pomdp,k]{V_{#1}^{\mathrm{restr}}}
\newcommand{\pDtmcInitMkSubs}[1][\pomdp,k]{\ensuremath{\dtmc_{#1}^{\mathrm{subs}} = (S_{#1},\sinitMk[#1],\VarMkSubs[#1],P_{#1}^{\mathrm{subs}})}}
\newcommand{\pDtmcInitMkNext}[1][\pomdp,k]{\ensuremath{\dtmc_{#1}^{\mathrm{next}} = (S_{#1},\allowbreak\sinitMk[#1],\allowbreak\VarMkNext[#1],\allowbreak P_{#1}^{\mathrm{next}})}}
\newcommand{\pDtmcInitMkRestr}[1][\pomdp,k]{\ensuremath{\dtmc_{#1}^{\mathrm{restr}} = (S_{#1},\allowbreak\sinitMk[#1],\allowbreak\VarMkRestr[#1],\allowbreak P_{#1}^{\mathrm{restr}})}}

\newcommand{\p}{\ensuremath{\mathbb{P}}}
\newcommand{\pr}{\ensuremath{\mathrm{Pr}}}

\newcommand{\R}{\mathbb{R}}
\newcommand{\Q}{\mathbb{Q}}

\newcommand{\N}{\mathbb{N}}

\newcommand{\FSCs}[2][]{\ensuremath{\mathit{FSC}_{#1}^{#2}}}
\newcommand{\PFSC}[3][]{\ensuremath{\mathit{PFSC}_{#1}^{#2}(#3)}}

\newcommand{\Ireal}{[0,\, 1]\subseteq\mathbb{R}}  

\newcommand{\Distr}{\mathit{Distr}}

\newcommand{\distDom}{X}

\newcommand{\distFunc}{\mu}
\newcommand{\distDomElem}{x}
\DeclareMathOperator{\supp}{supp}

\newcommand{\Until}{\mbox{$\, {\sf U}\,$}}

\newcommand{\Var}{\ensuremath{V}\xspace}        
\newcommand{\Paramvar}{\ensuremath{{V}}\xspace}        

\newcommand{\sinit}{s_{\mathrm{I}}} 

\newcommand{\ninit}{n_{\mathrm{I}}}
\newcommand{\mdp}{M}
\newcommand{\MdpInit}[1][]{\ensuremath{\mdp{#1}=(S{#1},\sinit,\Act,\probmdp{#1})}}
\newcommand{\pMdpInit}[1][]{\ensuremath{\mdp{#1}=(S{#1},\,\sinit{#1},\Act,\Var,\probmdp{#1})}}
\newcommand{\probmdp}{\mathcal{P}}
\newcommand{\probdtmc}{P}

\newcommand{\Strategy}{\sched} 
\newcommand{\DTMCgnd}{M^\Strategy}

\newcommand{\fsc}{\ensuremath{\mathcal{A}}}

\newcommand{\actionMap}{\ensuremath{\mathcal{\gamma}}}
\newcommand{\nodeTransition}{\ensuremath{\mathcal{\delta}}}
\newcommand{\FSCinit}{\fsc=(N,\ninit,\actionMap,\nodeTransition)}
\newcommand{\ObsSym}{{Z}}
\newcommand{\ObsFun}{{O}}
\newcommand{\obs}{\ensuremath{z}}
\newcommand{\PomdpInit}[1][]{\pomdp{#1}=(\mdp{#1},\ObsSym{#1},\ObsFun{#1})}
\newcommand{\pomdp}{\mathcal{M}}

\newcommand{\states}{\ensuremath{S}}

\newcommand{\poly}[1][]{\ensuremath{\mathbb{Q}[#1]}}

\newcommand{\pdtmc}{\ensuremath{\mathcal{P}}}

\newcommand{\sched}{\ensuremath{\sigma}}
\newcommand{\Sched}{\ensuremath{{\Sigma}}}
\newcommand{\osched}{\ensuremath{\mathit{\sigma}}}
\newcommand{\oSched}{\ensuremath{\Sigma}}

\newcommand{\Act}{\ensuremath{\mathit{Act}}}
\newcommand{\ActS}{\ensuremath{\mathit{A}}}
\newcommand{\act}{\ensuremath{a}}
\newcommand{\pmdp}{\ensuremath{M}}

\newcommand{\pathset}{\mathsf{Paths}}

\newcommand{\pathsfin}{\pathset_{\mathit{fin}}}

\newcommand{\last}[1]{\mathrm{last}(#1)}

\DeclareMathAlphabet{\mathpzc}{OT1}{pzc}{m}{it}
\def\presuper#1#2%
  {\mathop{}%
   \mathopen{\vphantom{#2}}^{#1}%
   \kern-\scriptspace%
   #2}

\theoremheaderfont{\bfseries}
\theorembodyfont{\normalfont}
\theoremseparator{:}
\theoremsymbol{$\blacksquare$}
\newtheorem*{proof}{Proof}

\begin{document}
\title{Permissive Finite-State Controllers of POMDPs \\ using Parameter Synthesis%
    \thanks{\,\,Supported by the DFG RTG 2236 ``UnRAVeL''.}
}
\author{Sebastian Junges\textsuperscript{1}, Nils Jansen\textsuperscript{2},
        Ralf Wimmer\textsuperscript{3}, Tim Quatmann\textsuperscript{1}, \\
        \textbf{Leonore Winterer\textsuperscript{3}, Joost-Pieter Katoen\textsuperscript{1},
        and Bernd Becker\textsuperscript{3}} \\
  \textsuperscript{1}RWTH Aachen University, Aachen, Germany \\
  \textsuperscript{2}Radboud University, Nijmegen, The Netherlands \\
  \textsuperscript{3}Albert-Ludwigs-Universit\"at Freiburg, Freiburg im Breisgau, Germany
}
\maketitle
\begin{abstract}
  We study finite-state controllers (FSCs) for partially observable Markov decision processes (POMDPs) that are provably correct with respect to given specifications.
  The key insight is that computing (randomised) FSCs on POMDPs is equivalent to---and computationally as hard as---synthesis for parametric Markov chains (pMCs).
  This correspondence allows to use tools for parameter synthesis in pMCs to compute correct-by-construction FSCs on POMDPs for a variety of specifications.
  Our experimental evaluation shows comparable performance to well-known POMDP solvers.
\end{abstract}

\section{Introduction}
\label{sec:introduction}

\paragraph{Partially Observable MDPs.} We intend to provide guarantees for planning scenarios given by dynamical systems with uncertainties.
In particular, we want to synthesise a \emph{strategy} for an agent that ensures certain desired behaviour~\cite{howard1960dynamic}.
A popular formal model for planning subject to stochastic behaviour are Markov decision processes (MDPs)~\cite{Put94}.
An MDP is a nondeterministic model in which the agent chooses to perform an action under full knowledge of the environment it is operating in.
The outcome of the action is a probability distribution over the system states.
Many applications, however, allow only \emph{partial observability} of the current system state~\cite{kaelbling1998planning,thrun2005probabilistic,WongpiromsarnF12,DBLP:books/daglib/0023820}.
For such applications, MDPs are extended to \emph{partially observable Markov decision processes} (POMDPs).
While the agent acts within the environment, it encounters certain \emph{observations}, according to which it can infer the likelihood of the system being in a certain state.
This likelihood is called the \emph{belief state}.
Executing an action leads to an update of the belief state according to new observations.
The belief state together with an update function form a (typically uncountably infinite) MDP, referred to as the \emph{belief MDP}~\cite{ShaniPK13}.

\paragraph{The POMDP Synthesis Problem.} For (PO)MDPs, a \emph{randomised strategy} is a function that resolves the nondeterminism by
providing a probability distribution over actions at each time step.
In general, strategies depend on the full history of the current evolution of the (PO)MDP.
If a strategy depends only on the current state of the system, it is called \emph{memoryless}.
For MDPs, memoryless strategies suffice to induce optimal values according to our measures of interest~\cite{Put94}.
Contrarily, POMDPs require strategies taking the full observation history into account~\cite{Ross83}, \eg in case of infinite-horizon objectives.
Moreover, strategies inducing \emph{optimal} values are computed by assessing the entire belief
MDP~\cite{MadaniHC99,braziunas2003pomdp,szer2005optimal,NPZ17}, rendering the problem undecidable~\cite{ChatterjeeCT16}.

POMDP strategies can be represented by \emph{infinite-state controllers}.
For computational tractability, strategies are often restricted to finite memory; this amounts to using \emph{randomised finite-state controllers} (FSCs)~\cite{meuleau1999solving}.
We often refer to strategies as FSCs.
Already the computation of a memoryless strategy adhering to a specification is NP-hard, SQRT-SUM-hard, and in PSPACE~\cite{VlassisLB12}.
While optimal values cannot be guaranteed, a small amount of memory in combination with \emph{randomisation} may
superseed large memory in many cases~\cite{chatterjee2004trading,amato2010optimizing}.

\paragraph{Correct-by-Construction Strategy Computation.}
In this paper, we synthesise FSCs for POMDPs.
We require these FSCs to be provably correct for specifications such as indefinite-horizon properties like expected reward or reach-avoid probabilities.
State-of-the-art POMDP solvers mainly consider expected discounted reward measures~\cite{DBLP:conf/aaai/WalravenS17}, which are a subclass of indefinite horizon properties~\cite{DBLP:conf/uai/KolobovMW12}.

Our key observation is that for a POMDP the \emph{set of all FSCs} with a fixed memory bound can be succinctly represented by a \emph{parametric Markov chain} (pMC)~\cite{Daw04}.
Transitions of pMCs are given by functions over a finite set of parameters rather than constant probabilities.
The \emph{parameter synthesis} problem for pMCs is to determine parameter instantiations that satisfy (or refute) a given specification.
We show that the pMC parameter synthesis problem and the POMDP strategy synthesis problem are equally hard.
This correspondence not only yields complexity results~\cite{DBLP:journals/corr/abs-1709-02093}, but particularly enables using a plethora of methods for parameter synthesis implemented in sophisticated and optimised parameter
synthesis tools like \tool{PARAM}~\cite{param_sttt}, \tool{PRISM}~\cite{KNP11}, and \tool{PROPhESY}~\cite{DJJ+15}.
As our experiments show, they turn out to be competitive alternatives to dedicated POMDP solvers.
Moreover, as we are solving slightly different problems, our methods are orthogonal to,
\eg, PRISM-POMDP~\cite{NPZ17} and solve-POMDP~\cite{DBLP:conf/aaai/WalravenS17}.

\begin{table}[tb]
\caption{Correspondence}
\centering\small
\label{tab:great_table_of_correspondence}
\begin{tabular}{|p{0.24\textwidth}l|}
\hline
POMDP under FSC        & pMC                      \\ \hline\hline
states $\times$ memory & states                   \\
same observation       & same parameter           \\
strategy               & parameter instantiation  \\
\hline
\end{tabular}
\end{table}
We detail our contributions and the structure of the paper, which starts with necessary formalisms in Sect.~\ref{sec:preliminaries}.

\begin{description}[labelindent=0cm,leftmargin=0cm, noitemsep,nolistsep]
\item[Section~\ref{sec:verification}:] We establish the correspondence of POMDPs and pMCs, see Tab.~\ref{tab:great_table_of_correspondence}.
The product of a POMDP and an FSC yields a POMDP with state-memory pairs, which we map to states in the pMC.
If POMDP states share \emph{observations}, the corresponding pMC states share \emph{parameters} at emanating transitions.
A \emph{strategy} of the POMDP corresponds to a \emph{parameter instantiation} in the pMC.

\item[Section~\ref{sec:equivalence}:] We show the opposite direction, namely a transformation from pMCs to POMDPs.
This result establishes that the (controller and parameter, respectively) synthesis problems for POMDPs and pMCs are equally hard.
Technically, we identify the practically relevant class of \emph{simple pMCs}, which coincides with POMDPs under memoryless strategies.

\item[Section~\ref{sec:equivalence:restrictions}:] Typical restrictions on parameter instantiations concern whether parameters may be assigned the probability zero.
We discuss effects of such restrictions to the resulting POMDP strategies.

\item[Section~\ref{sec:alternative_fscs}:] Specific types of FSCs differ in the information they take into account, \eg the last action that has been taken by an agent.
We compare existing definitions from the literature and discuss their effect in our setting.

\item[Section~\ref{sec:permissive}:] We lift the definition of permissive strategies to FSCs, and describe how the synthesis of an important class of permissive strategies is naturally cast into a well-known parameter synthesis problem.

\item[Section~\ref{sec:evaluation}:] We evaluate the computation of correct-by-construction FSCs using pMC synthesis techniques.
To that end, we explain how particular parameter synthesis approaches deliver optimal or near-optimal FSCs.
Then, we evaluate the approach on a range of typical POMDP benchmarks.
We observe that often a small amount of memory suffices. Our approach is competitive to state-of-the-art POMDP solvers and is able to synthesise small, almost-optimal FSCs.
\end{description}

\paragraph{Related Work.}
In addition to the cited works, \cite{meuleau1999solving} uses a branch-\&-bound method to find optimal FSCs for POMDPs.
A SAT-based approach computes FSCs for qualitative properties~\cite{DBLP:conf/aaai/ChatterjeeCD16}.
For a survey of decidability results and algorithms for broader classes of properties refer to~\cite{ChatterjeeCT16,ChatterjeeCGK16}.
Work on parameter synthesis~\cite{DBLP:journals/corr/abs-1709-02093,DBLP:conf/icse/FilieriGT11}
might contain additions to the methods considered here.

\section{Preliminaries}
\label{sec:preliminaries}
\noindent The set $[k]$ denotes $\{0, \hdots, k\} \subseteq \N$.

A \emph{probability distribution} over a finite or countably infinite set $\distDom$
is a function $\distFunc\colon\distDom\rightarrow\Ireal$ with $\sum_{\distDomElem\in\distDom}\distFunc(\distDomElem)=\distFunc(\distDom)=1$.
The set of all distributions on $\distDom$ is $\Distr(\distDom)$. The support of a distribution $\distFunc$ is
$\supp(\distFunc) = \{x\in\distDom\,|\,\distFunc(x)>0\}$.
A distribution is \emph{Dirac} if $|\!\supp(\distFunc)| = 1$.

Let $\Var=\{p_1,\ldots,p_n\}$ be a finite set of \emph{parameters} over the domain $\R$ and let $\poly[\Var]$ be the set of multivariate polynomials over $\Var$.
An \emph{instantiation} for $\Var$ is a function $u \colon \Var \to \R$.
Replacing each parameter $p$ in a polynomial $f \in \poly[\Var]$ by $u(p)$ yields $f[u] \in \R$.

Decision problems can be considered as languages describing all positive instances.
A language $L_1 \subseteq \{0,1\}^{*}$ is \emph{polynomial (many-one or Karp) reducible} to $L_2 \subseteq \{0,1\}^{*}$, written $L_1 \polred L_2$,
if there exists a polynomial-time computable function $f\colon \{0,1\}^{*} \rightarrow \{0,1\}^{*}$
such that for all $w \in \{0,1\}^{*}$, $w \in L_1 \iff f(w) \in L_2$.
Polynomial reductions are essential to define complexity classes, cf.~\cite{DBLP:books/daglib/0072413}.

\subsection{Parametric Markov Models}
\label{ssec:prob_models}

\begin{definition}[pMDP]
  \label{def:pmdp}
  A \emph{parametric Markov decision process} (pMDP) $\pmdp$ is a tuple $\pMdpInit$ with a finite
  (or countably infinite) set $\states$ of \emph{states}, \emph{initial state} $\sinit\in\states$, a finite set
  $\Act$ of \emph{actions}, a finite set $\Var$ of parameters, and a \emph{transition function}
  $\probmdp\colon \states\times\Act\times\states\rightarrow \poly[\Var]$.
\end{definition}
The \emph{available actions} in $s\in\states$ are $\ActS(s)=\{\act\in\Act\mid \exists s'\in\states: \probmdp(s,\act,s')\neq 0\}$.
W.\,l.\,o.\,g.\ we assume $\forall s,s'\in S.\,\forall \act \in \Act.\, \bigl(P(s,\act,s') \neq 0 \land P(s,\act,s') \neq 1\bigr)
\ \Rightarrow\ \bigl(\exists s''\neq s'.\,P(s,a,s'') \neq 0\bigr)$.
We assume that pMDP $\mdp$ contains no deadlock states, \ie $\ActS(s)\neq\emptyset$ for all $s\in\states$.
A \emph{path} of a pMDP $\mdp$ is an (in)finite sequence $\pi = s_0\xrightarrow{\act_0}s_1\xrightarrow{\act_1}\cdots$,
where $s_0=\sinit$, $s_i\in\states$, $\act_i\in\ActS(s_i)$, and $\probmdp(s_i,\act_i,s_{i+1})\neq 0$ for all $i\in\N$.
For finite $\pi$, $\last{\pi}$ denotes
the last state of $\pi$. The set of (in)finite paths of $\mdp$ is $\pathsfin^{\mdp}$ ($\pathset^{\mdp}$).
\begin{definition}[MDP]
  \label{def:mdp}
  A \emph{Markov decision process (MDP)} is a pMDP where $\probmdp \colon\states \times \Act \times \states \rightarrow [0,1] \subseteq \R$
  and  for all $s \in \states$ and $\act \in \ActS(s)$, $\sum_{s'\in\states}\probmdp(s,\act,s') = 1$.
\end{definition}

A \emph{(parametric) discrete-time Markov chain} ((p)MC) is a (p)MDP with $|\ActS(s)|=1$ for all $s\in\states$.
For a (p)MC $\dtmc$, we may omit the actions and use the notation $\pDtmcInit$ with a transition function $\probdtmc$
of the form $\probdtmc\colon\states\times\states\rightarrow \poly[\Var]$.

Applying an \emph{instantiation} $u\colon\Var\to\R$ to a pMDP or pMC $\mdp$, denoted $\mdp[u]$, replaces each polynomial
$f$ in $\mdp$ by $f[u]$. $\mdp[u]$ is also called the \emph{instantiation} of $\mdp$ at $u$.
Instantiation $u$ is \emph{well-defined} for $\mdp$ if the replacement yields probability distributions, \ie
if $\mdp[u]$ is an MDP or an MC, respectively.

\paragraph{Strategies.}
To resolve the nondeterministic action choices in MDPs, so-called \emph{strategies} determine at each state a distribution over actions to take.
This decision may be based on the \emph{history} of the current path.
\begin{definition}[Strategy]
  \label{def:strategy}
  A \emph{strategy} $\sched$ for (p)MDP $\mdp$ is a function $\sched\colon \pathsfin^{\mdp}\to\Distr(\Act)$
  such that $\supp\bigl(\sched(\pi)\bigr) \subseteq \Act\bigl(\last{\pi}\bigr)$ for all $\pi\in \pathsfin^{\mdp}$.
  The set of all strategies of $\mdp$ is $\Sched^{\mdp}$.
\end{definition}
A strategy $\sched$ is \emph{memoryless} if $\last{\pi}=\last{\pi'}$ implies $\sched(\pi)=\sched(\pi')$ for all $\pi,\pi'\in\pathsfin^{\mdp}$.
It is \emph{deterministic} if $\sched(\pi)$ is a Dirac distribution for all $\pi\in\pathsfin^{\mdp}$.
A strategy that is not deterministic is \emph{randomised}.

A strategy~$\sigma$ for an MDP $\mdp$ resolves all nondeterministic choices, yielding an \emph{induced Markov chain} $\DTMCgnd$,
for which a \emph{probability measure} over infinite paths is defined by the cylinder set construction~\cite{BK08}.

\begin{definition}[Induced Markov Chain]
  \label{def:induced_dtmc}
 For an MDP $\MdpInit$ and a strategy $\sched\in\Sched^{\mdp}$,  the MC \emph{induced by  $\mdp$ and $\sched$} is given by $\DTMCgnd = (\pathsfin^{\mdp},\sinit,\probdtmc^{\sched})$ where:
  \[
    \probdtmc^{\sched}(\pi,\pi') = \begin{cases}
        \probmdp(\last{\pi},\act,s')\cdot\sched(\pi)(\act) & \text{if $\pi' = \pi\act s'$,} \\
        0 & \text{otherwise.}
      \end{cases}
  \]
\end{definition}

\subsection{Partial Observability}
\label{ssec:partial_obs}
\begin{definition}[POMDP]
  \label{def:pomdp}
  A \emph{partially observable MDP (POMDP)} is a tuple $\PomdpInit$, with $\MdpInit$ the \emph{underlying MDP of $\pomdp$}, $\ObsSym$ a finite set of observations and $\ObsFun\colon\states\rightarrow\ObsSym$ the \emph{observation function}.
\end{definition}
We require that states with the same observations have the same set of
enabled actions, \ie $\ObsFun(s)=\ObsFun(s')$ implies $\ActS(s)=\ActS(s')$ for all $s,s'\in\states$.
We define $\ActS(z) = \ActS(s)$ if $\ObsFun(s) = z$.
More general observation functions~\cite{RoyGT05,ShaniPK13} take the last action into account and provide a distribution
over $\ObsSym$. There is a transformation of the general case to the POMDP definition used here
that blows up the state space polynomially~\cite{ChatterjeeCGK16}.
In Fig.~\ref{fig:pomdpfragment}, a fragment of the underlying MDP of a POMDP has two different observations,
indicated by the state colouring.

We lift the observation function to paths: For
$\pi=s_0\xrightarrow{\act_0} s_1\xrightarrow{\act_1}\cdots s_n\in\pathsfin^{\mdp}$, the associated
\emph{observation sequence} is $\ObsFun(\pi)=\ObsFun(s_0)\xrightarrow{\act_0} \ObsFun(s_1)\xrightarrow{\act_1}\cdots\ObsFun(s_n)$.
Several paths in the underlying MDP may yield the same observation sequence.
Strategies have to take this restricted observability into account.
\begin{definition}
  \label{def:obsstrategy}
  An \emph{observation-based strategy} $\osched$ for a POMDP $\pomdp$ is a strategy
  for the underlying MDP $\mdp$
  such that $\osched(\pi)=\osched(\pi')$ for all $\pi,\pi'\in\pathsfin^{\mdp}$
  with $\ObsFun(\pi)=\ObsFun(\pi')$.
  $\oSched^\pomdp$ is the set of observation-based strategies for $\pomdp$.
\end{definition}
An observation-based strategy selects actions based on observations along a path and the past actions.
Applying the strategy to a POMDP yields an induced MC as in Def.~\ref{def:induced_dtmc}, resolving all nondeterminism and partial observability.
To represent observation-based strategies with finite memory, we define \emph{finite-state controllers} (FSCs).
We discuss alternative definitions from the literature in Sect.~\ref{sec:alternative_fscs}.
A randomised observation-based strategy for a POMDP $\pomdp$ with (finite) $k$ memory is represented by an FSC $\fsc$ with $k$ memory nodes.
If $k=1$, the FSC describes a \emph{memoryless strategy}.
We often refer to observation-based strategies as FSCs.
\begin{figure}[tb]
  \centering
  \subfigure[POMDP $\pomdp$]{
    \label{fig:pomdpfragment}
    \scalebox{0.65}{
      \begin{tikzpicture}
\node[state] (s1) {$s_1$};
\node[fill, circle, inner sep=2pt, right=0.4cm of s1] (a1) {};
\node[fill, circle, inner sep=2pt, left=0.4cm of s1] (a2) {};

\node[state,right=1.2cm of s1, yshift=0.5cm, fill=red!40] (s2) {$s_2$};
\node[state,right=1.2cm of s1, yshift=-0.5cm] (s3) {$s_3$};

\node[state,left=1.2cm of s1, yshift=-0.5cm, fill=red!40] (s4) {$s_4$};
\node[state,left=1.2cm of s1, yshift=0.5cm, fill=red!40] (s5) {$s_5$};
\draw ($(s1.north) + (0,0.3)$) edge[->] (s1);
\draw[] (s1) edge node[above] {$\act_1$} (a1);

\draw[->] (a1) edge node[above] {\scriptsize$0.6$} (s2);
\draw[->] (a1) edge node[below] {\scriptsize$0.4$} (s3);

\draw[] (s1) edge node[above] {$\act_2$} (a2);

\draw[->] (a2) edge node[below] {\scriptsize$0.7$} (s4);
\draw[->] (a2) edge node[above] {\scriptsize$0.3$} (s5);
\end{tikzpicture}
    }
  } \\
  \subfigure[FSC $\fsc$]{
    \label{fig:fsc_for_pomdp}
    \scalebox{0.5}{
      \begin{tikzpicture}
	\node[state] (s1) {$\langle n_1 \rangle$} ;
	
	\node[state,below=3.2cm of s1] (s3) {$\langle n_2\rangle$} ;
	\node[fill, circle, inner sep=2pt, below=of s1, xshift=-0.6cm] (n1q1) {};
 	
	\node[fill, circle, inner sep=2pt, below=of s1, xshift=0.6cm] (n1q2) {};
 	
	\node[draw,rectangle, below=of n1q1, xshift=-1.5cm] (q1a1) {$\act_1$};
	\node[draw,rectangle, below=of n1q1] (q1a2) {$\act_2$};
	\node[draw,rectangle, below=of n1q2] (q2a1) {$\act_1$};
	\node[draw,rectangle, below=of n1q2, xshift=1.5cm] (q2a2) {$\act_2$};
	
	\node[fill, circle, inner sep=2pt, left=0.5cm of q1a1] (q1a1dist) {};
	
	\node[fill, circle, inner sep=2pt, left=0.5cm of q1a2] (q1a2dist) {};

	\node[fill, circle, inner sep=2pt, right=0.5cm of q2a1] (q2a1dist) {};

	\node[fill, circle, inner sep=2pt, right=0.5cm of q2a2] (q2a2dist) {};
	\draw ($(s1.north) + (0,0.3)$) edge[->] (s1);
	
	\draw[] (s1) edge node[left] {$\obs_0$?} (n1q1);
	\draw[] (s1) edge node[right] {$\obs_1$?} (n1q2);
	
	\draw[->] (n1q1) edge node[auto] {} (q1a1);
	\draw[->] (n1q1) edge node[auto] {} (q1a2);
	\draw[->] (n1q2) edge node[auto] {} (q2a1);
	\draw[->] (n1q2) edge node[auto] {} (q2a2);
	
	\draw[] (q1a1) edge (q1a1dist);
	\draw[] (q1a2) edge (q1a2dist);
	\draw[] (q2a1) edge (q2a1dist);
	\draw[] (q2a2) edge (q2a2dist);
	
	\draw[->] (q1a1dist) edge[bend left=50] node[auto] {} (s1);
	\draw[->] (q1a1dist) edge[bend right=40] node[auto] {} (s3);
	
	\draw[->] (q1a2dist) edge[bend left=50] node[auto] {} (s1);
	\draw[->] (q1a2dist) edge[bend right=40] node[auto] {}  (s3);
	
	\draw[->] (q2a1dist) edge[bend right=50] node[auto] {}  (s1);
	\draw[->] (q2a1dist) edge[bend left=40] node[auto] {}  (s3);
	
	\draw[->] (q2a2dist) edge[bend right=50] node[auto] {}  (s1);
	\draw[->] (q2a2dist) edge[bend left=40] node[auto] {}  (s3);
	
	\end{tikzpicture}
    }
  }
  \subfigure[Induced MC $\pomdp^{\osched_\fsc}$]{
    \label{fig:induced_mc_fsc_pomdp}
    \scalebox{0.6}{
      \begin{tikzpicture}
\node[] (s1) {$\langle s_1,n_1 \rangle$};
 
\node[right=1cm of s1, yshift=1.8cm] (s2n1) {$\langle s_2,n_1 \rangle$};
\node[right=1cm of s1, yshift=0.6cm] (s2n2) {$\langle s_2,n_2 \rangle$};
\node[right=1cm of s1, yshift=-0.6cm] (s3n1) {$\langle s_3,n_1 \rangle$};
\node[right=1cm of s1, yshift=-1.8cm] (s3n2) {$\langle s_3,n_2 \rangle$};

\node[left=1cm of s1, yshift=1.8cm] (s5n1) {$\langle s_5,n_1 \rangle$};
\node[left=1cm of s1, yshift=0.6cm] (s5n2) {$\langle s_5,n_2 \rangle$};
\node[left=1cm of s1, yshift=-0.6cm] (s4n1) {$\langle s_4,n_1 \rangle$};
\node[left=1cm of s1, yshift=-1.8cm] (s4n2) {$\langle s_4,n_2 \rangle$};

\draw ($(s1.north) + (0,0.3)$) edge[->] (s1);

\draw[->] (s1) edge node[above=0.3cm] {\scriptsize$0.15$} (s2n1);
\draw[->] (s1) edge node[above] {\scriptsize$0.15$} (s2n2);
\draw[->] (s1) edge node[below] {\scriptsize$0.1$} (s3n1);
\draw[->] (s1) edge node[below=0.3cm] {\scriptsize$0.1$} (s3n2);

\draw[->] (s1) edge node[above=0.3cm] {\scriptsize$0.075$} (s5n1);
\draw[->] (s1) edge node[above] {\scriptsize$0.075$} (s5n2);
\draw[->] (s1) edge node[below] {\scriptsize$0.175$} (s4n1);
\draw[->] (s1) edge node[below=0.3cm] {\scriptsize$0.175$} (s4n2);

\end{tikzpicture}
    }
  }
  \caption{(a)~The POMDP~$\pomdp$ with observations $\ObsFun(s_1)=\ObsFun(s_3)=\obs_0$ (white)
    and $\ObsFun(s_2)=\ObsFun(s_4)=\ObsFun(s_5)=\obs_1$ (red).
    (b)~The associated (partial) FSC~$\fsc$ has two memory nodes.
    (c)~A part of MC~$\pomdp^{\osched_\fsc}$ induced by $\pomdp$ and $\fsc$.
  }
  \label{fig:pomdp_and_fsc}
\end{figure}
\begin{definition}[FSC]
  \label{def:fsc}
  A \emph{finite-state controller (FSC)} for a POMDP $\pomdp$ is a tuple $\FSCinit$, where $N$ is a finite set of \emph{memory nodes},
  $\ninit\in N$ is the \emph{initial memory node}, $\actionMap$ is the \emph{action mapping} $\actionMap\colon N\times\ObsSym\rightarrow\Distr(\Act)$,
  and $\nodeTransition$ is the \emph{memory update} $\nodeTransition\colon N\times\ObsSym\times\Act\rightarrow \Distr(N)$.
  The set $\FSCs[k]{\pomdp}$ denotes the set of FSCs with $k$ memory nodes, called \emph{$k$-FSC}s.
  Let $\osched_\fsc\in\oSched^\pomdp$ denote the observation-based strategy represented by $\fsc$.
\end{definition}
From a node $n$ and the observation $\obs$ in the current state of the POMDP, the next action $\act$ is
chosen from $\ActS(\obs)$ randomly as given by $\actionMap(n, \obs)$. Then the successor node of the FSC is determined randomly via $\nodeTransition(n, \obs, \act)$.

\begin{example}
  \label{ex:constant_induced_mc}
  Fig.~\ref{fig:fsc_for_pomdp} shows an excerpt of an FSC $\fsc$
  with two memory nodes.
  From node $n_1$, the action mapping distinguishes observations $\obs_0$ and $\obs_1$.
  The solid dots indicate a probability distribution from $\Distr(\Act)$.
  For readability, all distributions are uniform and we omit the action mapping for node $n_2$.

  Now recall the POMDP $\pomdp$ from Fig.~\ref{fig:pomdpfragment}.
  The induced MC~$\pomdp^{\osched_\fsc}$ is shown in Fig.~\ref{fig:induced_mc_fsc_pomdp}.
  Assume $\pomdp$ is in state $s_1$ and $\fsc$ in node $n_1$.
  Based on the observation $\obs_0\colonequals\ObsFun(s_1)$, $\osched_\fsc$ chooses action $\act_1$ with probability $\nodeTransition(n_1,\obs_0)(\act_1)=0.5$  leading to the probabilistic branching in the POMDP.
  With probability $0.6$, $\pomdp$ evolves to state $s_2$.
  Next, the FSC $\fsc$ updates its memory node; with probability $\nodeTransition(n_1, \obs_0, \act_1)(n_1)=0.5$, $\fsc$ stays in $n_1$.
  The corresponding transition from $\langle s_1,n_1\rangle$ to $\langle s_2,n_1\rangle$
  in  $\pomdp^{\osched_\fsc}$ has probability $0.5\cdot 0.6\cdot 0.5=0.15$.
\end{example}

\subsection{Specifications}
\label{ssec:spec}

For a POMDP $\pomdp$, a set $G\subseteq\states$ of \emph{goal states}, a set $B\subseteq\states$ of \emph{bad states}, and a \emph{threshold} $\lambda \in [0,1)$, we consider \emph{quantitative reach-avoid specifications}
$\varphi=\p_{> \lambda} (\neg B\,\Until\,G)$. The specification $\varphi$ is satisfied for a
strategy $\osched\in\oSched^\pomdp$ if the probability $\pr^{\pomdp^\osched}(\neg B\,\Until\,G)$ of reaching a goal state in $\pomdp^\osched$
without entering a bad state in between exceeds $\lambda$, denoted by $\pomdp^\osched\models\varphi$.
The task is to compute such a strategy provided that one exists. For an MDP $\mdp$,
there is a memoryless deterministic strategy inducing the maximal probability $\pr_{\max}^\mdp(\neg B\,\Until\,G)$~\cite{Condon92}.
For a POMDP $\pomdp$, however, observation-based strategies with infinite memory as in Def.~\ref{def:obsstrategy} are
necessary~\cite{Ross83} to attain $\pr_{\max}^\pomdp(\neg B\,\Until\,G)$.
The problem of proving the satisfaction of $\varphi$ is therefore undecidable~\cite{ChatterjeeCT16}.
In our experiments, we also use \emph{undiscounted expected reachability reward properties}~\cite{BK08}.

\section{From POMDPs to pMCs}
\label{sec:verification}
Our goal is to make pMC synthesis methods available for POMDPs.
In this section we provide a transformation from a POMDP $\pomdp$ to a pMC $\dtmc$.
We consider the following decision~problems.
\begin{problem}[\textsl{$\exists$k-FSC}]
Given a POMDP $\pomdp$, a specification $\varphi$, and a (unary encoded) memory bound $k>0$, is there a $k$-FSC $\fsc$ with $\pomdp^{\osched_\fsc} \models \varphi$?
\end{problem}
\begin{problem}[\textsl{$\exists$INST}]
	For a pMC $\dtmc$ and a specification $\varphi$,
	does a well-defined instantiation $u$ exist such that $\dtmc[u]  \models \varphi$?
\end{problem}
\begin{theorem}
	$\exists k$-FSC $\polred$  $\exists$INST.
\end{theorem}

The remainder of the section outlines the proof,
the converse direction is addressed in Sect.~\ref{sec:equivalence}. Consider a POMDP $\pomdp$, a specification $\varphi$, and a memory bound $k>0$ for which $\exists k$-FSC is to be solved.
The degrees of freedom to select a $k$-FSC are given by the possible choices for the action-mapping $\actionMap$ and the memory update $\nodeTransition$.
For each $\actionMap$ and $\nodeTransition$, we get a different induced MC, but these MCs are \emph{structurally similar} and can be represented by a single pMC.

\begin{figure}[t]
  \centering
  \subfigure[\small Induced pMC]{%
    \label{fig:example_fsc:applied_our_current}
    \scalebox{0.9}{\begin{tikzpicture}
  [every node/.append style={font=\scriptsize}]
\node[] (s1) {$\langle s_1, n_1 \rangle$};
\node[fill, circle, inner sep=2pt, right=0.8cm of s1] (a1) {};
\draw[->] (s1) edge node[above, pos=0.4] {$p$} (a1);

\node[fill, rectangle, inner sep=2pt,above=0.9cm of a1] (x1) {};

\node[fill, rectangle, inner sep=2pt,below=0.9cm of a1] (x2) {};

\draw[->] (a1) edge node[left, pos=0.6] {$0.6$} (x1);
\draw[->] (a1) edge node[left, pos=0.6] {$0.4$} (x2);

\draw ($(s1.west) + (-0.3,0)$) edge[->] (s1);

	\node[right=2cm of s1, yshift=1.5cm] (s2) {$\langle s_2, n_1 \rangle$};
	
	\node[right=2cm of s1, yshift=0.5cm] (s3) {$\langle s_2, n_2 \rangle$};
	\node[right=2cm of s1, yshift=-0.5cm] (s4) {$\langle s_3, n_1 \rangle$};
	
	\node[right=2cm of s1, yshift=-1.5cm] (s5) {$\langle s_3, n_2 \rangle$};
	
\draw[->] (x1) edge node[above] {$q_1$} (s2);
\draw[->] (x1) edge node[below] {$1{-}q_1$} (s3);

\draw[->] (x2) edge node[above] {$q_1$} (s4);
\draw[->] (x2) edge node[below] {$1{-}q_1$} (s5);
\end{tikzpicture}}
  }
  \subfigure[\small Parameterised transition probabilities]{%
    \scalebox{0.9}{\scriptsize
\begin{tabular}[b]{@{}|l|l|l|l|}
\hline
$\Act$                        & $\probmdp$                 & Node   & Result \\ \hline\hline
\multirow{4}{*}{$\act_1\colon p$}  & \multirow{2}{*}{0.6} & $n_1\colon q_1$ &  $0.6 \cdot p \cdot q_1$\\ \cline{3-4} 
                     &                     & $n_2\colon 1-q_1$ 				&   $0.6 \cdot p \cdot (1-q_1)$    	\\ \cline{2-4} 
                     & \multirow{2}{*}{0.4} & $n_1\colon q_1$ 				&  $0.4 \cdot p \cdot q_1$     		\\ \cline{3-4} 
                     &                      & $n_2\colon 1-q_1$ 				&  $0.4 \cdot p \cdot (1-q_1)$ 		\\ \cline{1-4} 
\multirow{4}{*}{$\act_2\colon 1-p$} & \multirow{2}{*}{0.7} & $n_1\colon q_2$ 			&  $0.7 \cdot (1-p) \cdot q_2$  		\\ \cline{3-4} 
                     &                      & $n_2\colon 1-q_2$ 				& $0.7 \cdot (1-p) \cdot (1-q_2)$	\\ \cline{2-4} 
                     & \multirow{2}{*}{0.3} & $n_1\colon q_2$ 				&  $0.3 \cdot (1-p) \cdot q_2$      	\\ \cline{3-4} 
                     &                      & $n_2\colon 1-q_2$ 				&  $0.3 \cdot (1-p) \cdot (1-q_2)$ \\ \hline
\end{tabular}
}
    \label{tab:fsc_example_tab}
  }
\caption{Induced parametric Markov chain for FSCs.}
\label{fig:induced_pmc_example}
\end{figure}
\begin{example}
  \label{ex:fsctopmc}
  Recall Fig.~\ref{fig:pomdp_and_fsc} and Ex.~\ref{ex:constant_induced_mc}.
  The action mapping $\actionMap$ and the memory update $\nodeTransition$ have arbitrary but fixed probability distributions.
  For $\act_1$, we represent the probability $\actionMap(n_1, \obs_0)(\act_1)\equalscolon p$ by $p\in[0,1]$.
  The memory update yields
  $\nodeTransition(n_1, \obs_0, \act_1)(n_1)\equalscolon q_1\in[0,1]$ and $\nodeTransition(n_1, \obs_0, \act_1)(n_2)\equalscolon 1-q_1$, respectively.
  Fig.~\ref{fig:example_fsc:applied_our_current} shows the induced pMC for action choice $\act_1$.
  For instance, the transition from $\langle s_1,n_1\rangle$ to $\langle s_2,n_1\rangle$ is labelled with polynomial $p\cdot 0.6\cdot q_1$. \\
  We collect all polynomials for observation $\obs_0$ in Fig.~\ref{tab:fsc_example_tab}.
  The \emph{result} column describes a \emph{parameterised distribution} over tuples of states  and memory nodes.
  Thus, instantiations for these polynomials need to sum up to one.
\end{example}
As the next step, we define the pMC that results from combining a $k$-FSC with a POMDP.
The idea is to assign parameters as arbitrary probabilities to action choices.
Each observation has one  \emph{remaining action}  given by a mapping $\remaining\colon \ObsSym \rightarrow \Act$.
$\remaining(\obs) \in \ActS(\obs)$ is the action to which, after choosing probabilities
for all other actions in $\ActS(\obs)$, the remaining probability is assigned.
A similar principle holds for the remaining memory node.

\begin{definition}[Induced pMC for a {\boldmath $k$}-FSC on POMDPs]
    \label{def:pomdp_to_pmc_direct}
    Let $\PomdpInit$ be a POMDP with $\MdpInit$ and let $k>0$ be a memory bound.
    The \emph{induced pMC} \allowbreak$\pDtmcInitMk$ is defined by:
    \begin{compactitem}
    \item $\states_{\pomdp,k} = \states \times \{0, \hdots,k-1\}$, 
    \item $\sinitMk = \langle \sinit,0 \rangle$,
    \item $V_{\pomdp,k} = \bigl\{ p^{\obs,n}_\act ~\big|~ \obs \in \ObsSym, n \in [k{-}1],$ \\
            \hspace*{\fill} $\act \in \ActS(\obs), \act \neq \remaining(\obs) \bigr\}$ \\
            $\phantom{V_{\pomdp,k}}\, \cup \bigl\{ q^{\obs,n}_{\act,n'} ~\big|~\obs \in \ObsSym,  n,n' \in [k{-}1],$\\
            \hspace*{\fill} $n' \neq k-1, \act \in \ActS(\obs) \bigr\}$,
    \item $\probdtmc_{\pomdp,k}(s,s') = \sum_{\act \in \ActS(s)} H(s,s',\act)$ for all $s,s'\in S'$,
    \end{compactitem}
    where $H\colon\states_{\pomdp,k} \times\states_{\pomdp,k} \times \Act\to\R$ is for
    $\obs = \ObsFun(s)$ defined by $H\bigl(\langle s, n \rangle, \langle s', n' \rangle, \act\bigr)=$
    \begin{align*}
        \probmdp(s,\act,s') & \cdot \left\{\begin{array}{lr}
        p^{z,n}_\act, & \text{if } \act \neq \remaining\bigl(z\bigr)\\
        1-\sum\limits_{b\neq \act}p^{z,n}_{b}, & \text{if } \act = \remaining\bigl(z\bigr)
        \end{array}\right\} \\
            & \cdot \left\{\begin{array}{lr}
        q^{z,n}_{\act,n'}, & \text{if } n'\neq k{-}1\\
        1-\sum\limits_{\bar{n}\neq n'}q^{z,n}_{\act,\bar{n}}, & \text{if } n'=k{-}1
        \end{array}\right\}
        \!\!\! \begin{array}{l} \\.\end{array}
    \end{align*}
\end{definition}
Intuitively, $H(s,s',\act)$ describes the probability mass from $s$ to $s'$ in the induced pMC that is
contributed by action $\act$.
The three terms correspond to the terms as seen in the first three columns of Fig.~\ref{tab:fsc_example_tab}.
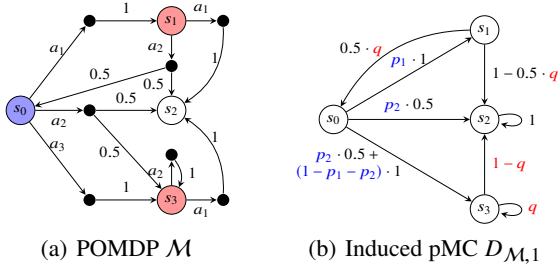
\begin{figure}[t]
\centering
\hspace*{\stretch{1}}
\subfigure[POMDP $\pomdp$]{
\label{fig:POMDPtoAnyPMC:pomdp}
\hspace*{\stretch{1}}
\scalebox{0.8}{
\begin{tikzpicture}[font=\scriptsize]
\node[state, fill=blue!40] (s0) {$s_0$};
\node[state, right=2cm of s0, fill=white!40] (s2)  {$s_2$};
\node[state, above=of s2, fill=red!40] (s1)  {$s_1$};
\node[state, below=of s2, fill=red!40] (s3)  {$s_3$};

\node[circle, inner sep=2pt, fill=black, left=1cm of s1] (a1) {};
\node[circle, inner sep=2pt, fill=black, left=1cm of s2] (a2) {};
\node[circle, inner sep=2pt, fill=black, left=1cm of s3] (a3) {};
\node[circle, inner sep=2pt, fill=black, right=0.5cm of s1] (a4) {};
\node[circle, inner sep=2pt, fill=black, below=0.4cm of s1] (a5) {};
\node[circle, inner sep=2pt, fill=black, right=0.5cm of s3] (a6) {};
\node[circle, inner sep=2pt, fill=black, above=0.4cm of s3] (a7) {};

\draw[->] (s0) -- node[above] {$\act_1$} (a1);
\draw[->] (s0) -- node[below] {$\act_2$} (a2);
\draw[->] (s0) -- node[above] {$\act_3$} (a3);
\draw[->] (s1) -- node[above] {$\act_1$} (a4);
\draw[->] (s1) -- node[left] {$\act_2$} (a5);
\draw[->] (s3) -- node[below] {$\act_1$} (a6);
\draw[->] (s3) -- node[left] {$\act_2$} (a7);

\draw[->] (a1) -- node[above] {$1$} (s1);
\draw[->] (a2) -- node[above] {$0.5$} (s2);
\draw[->] (a2) -- node[left] {$0.5$} (s3);
\draw[->] (a3) -- node[above] {$1$} (s3);

\draw[->] (a4) edge[bend left] node[above] {$1$} (s2);

\draw[->] (a5) -- node[above] {$0.5$} (s0);
\draw[->] (a5) -- node[left] {$0.5$} (s2);

\draw[->] (a6) edge[bend right] node[above] {$1$} (s2);

\draw[->] (a7) edge[bend left] node[right] {$1$} (s3);
\end{tikzpicture}
}
}\hspace*{\stretch{1}}
\subfigure[Induced pMC $\MC_{\pomdp,1}$]{
\label{fig:POMDPtoAnyPMC:pmc}
\quad\scalebox{0.8}{
\begin{tikzpicture}[font=\scriptsize]
\node[state] (s0) {$s_0$};
\node[state, right=2cm of s0] (s2)  {$s_2$};
\node[state, above=of s2] (s1)  {$s_1$};
\node[state, below=of s2] (s3)  {$s_3$};

\draw[->] (s0) -- node[above] {${\color{blue}p_1} \cdot 1$} (s1);
\draw[->] (s0) -- node[above] {${\color{blue}p_2} \cdot 0.5$} (s2);
\draw[->] (s0) -- node[left, align=center] {${\color{blue}p_2} \cdot 0.5+{}$\\${\color{blue}(1-p_1-p_2)} \cdot 1 $} (s3);

\draw[->] (s3) edge[loop right] node[right] {$\color{red}q$} (s3);
\draw[->] (s3) edge node[right] {$\color{red}1-q$} (s2);
\draw[->] (s1) edge[bend right] node[left] {$0.5\cdot {\color{red}q}$\ \ } (s0);
\draw[->] (s1) edge node[right] {$1-0.5\cdot {\color{red}q}$} (s2);
\draw[->] (s2) edge[loop right] node[right] {$1$} (s2);
\end{tikzpicture}\quad
}
}
\hspace*{\stretch{1}}
\caption{From POMDPs to pMCs ($k=1$)}
\label{fig:POMDPtoAnyPMC}
\end{figure}
\begin{example}
  Consider the POMDP in Fig.~\ref{fig:POMDPtoAnyPMC:pomdp} and let $k=1$.
  The induced pMC is given in Fig.~\ref{fig:POMDPtoAnyPMC:pmc}.
  The three actions from $s_0$ have probability $p_1$, $p_2$, and $1{-}p_1{-}p_2$ for the \emph{remaining action} $\act_3$.
  From the indistinguishable states $s_1$, $s_3$, actions have probability $q$ and $1{-}q$, respectively.
\end{example}
By construction, the induced pMC describes the set of all induced MCs:
\begin{theorem}[Correspondence Theorem]\label{theorem:great_theorem_of_correspondence}
  For POMDP $\pomdp$, memory bound $k$, and the induced pMC $\MC_{\pomdp,k}$:
  \[
    \bigl\{ \MC_{\pomdp,k}[u]~\big|~ u \text{ well-defined} \bigr\} = \bigl\{ \pomdp^{\osched_\fsc}~\big|~\fsc \in \FSCs[k]{\pomdp} \bigr\}.
  \]
  In particular, \emph{every well-defined instantiation $u$ describes an FSC $\fsc_u \in \FSCs[k]{\pomdp}$}.
\end{theorem}
By the correspondence, we can thus evaluate an instantiation of the induced pMC to assess whether the
corresponding $k$-FSC satisfies a given specification.
\begin{corollary}
  Given an induced pMC  $\MC_{\pomdp,k}$ and a specification $\varphi$: For every well-defined instantiation $u$ of $\MC_{\pomdp,k}$  and the corresponding $k$-FSC $\fsc_u$ we have: \\
  \centerline{$\pomdp^{\osched_{\fsc_u}} \models \varphi \iff \MC_{\pomdp,k}[u]\models \varphi$.}
\end{corollary}

\iftoggle{TR}{
\begin{lemma}[Number of Parameters]
  The number of parameters in the induced pMC $\dtmc_{\pomdp, k}$ is given by
  $\mathcal{O}\bigl(|\ObsSym| \cdot k^2 \cdot \max_{\obs \in \ObsSym} |\ActS(\obs)|\bigr)$.
\end{lemma}}{
 The number of parameters in the induced pMC $\dtmc_{\pomdp, k}$ is given by
  $\mathcal{O}\bigl(|\ObsSym| \cdot k^2 \cdot \max_{\obs \in \ObsSym} |\ActS(\obs)|\bigr)$.
}

\section{From pMCs to POMDPs (and Back Again)}
\label{sec:equivalence}
In the previous section we have shown that \textsl{$\exists$k-FSC} is at least as hard as \textsl{$\exists$INST}.
We now discuss whether both problems are equally hard:
The open question is whether we can reduce \textsl{$\exists$INST} to \textsl{$\exists$k-FSC}.

A straightforward reduction maintains the states of the pMC in the POMDP, or even yields a POMDP with the same graph structure (the topology) as the pMC.
The next example shows that this naive reduction is impossible.
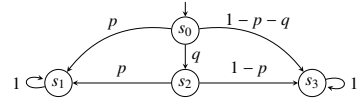
\begin{figure}[t]
	\centering
		\scalebox{0.65}{
			\begin{tikzpicture}
\node[state] (s0) {$s_0$};
\node[state, below=0.5cm of s0] (s2)  {$s_2$};
\node[state, left=2cm of s2] (s1)  {$s_1$};
\node[state, right=2cm of s2] (s3)  {$s_3$};

\draw ($(s0.north) + (0,0.3)$) edge[->] (s0);
\draw[->] (s0) edge[bend right] node[above] {$p$} (s1);
\draw[->] (s0) -- node[right] {$q$} (s2);
\draw[->] (s0) edge[bend left] node[above] {$1-p-q$} (s3);
\draw[->] (s2) -- node[above] {$p$} (s1);
\draw[->] (s2) -- node[above] {$1-p$} (s3);
\path[->] (s1) edge  [loop left] node {$1$} ();
\path[->] (s3) edge  [loop right] node {$1$} ();
\end{tikzpicture}
		}
		\caption{Non-simple pMC}
		\label{fig:PMCtoPOMDPIssue}
\end{figure}

\begin{example}
  In the pMC in Fig.~\ref{fig:PMCtoPOMDPIssue} the parameter $p$ occurs in two different distributions (at $s_0$ and $s_2$).
  For defining a reduction where the resulting POMDP has the same set of states, there are two options for the observation function at the states $s_0$ and $s_2$:
  Either $\ObsFun(s_0) = \ObsFun(s_2)$ or $\ObsFun(s_0) \neq \ObsFun(s_2)$.
  The intuition is that every (parametric) transition in the pMC corresponds to an action choice in a POMDP.
  Then $\ObsFun(s_0) = \ObsFun(s_2)$ is impossible as $s_0$ and $s_2$ have a different number of outgoing transitions (outdegree).
  Adding a self-loop to $s_2$ does not alleviate the problem.
  Moreover, $\ObsFun(s_0) \neq \ObsFun(s_2)$ is impossible, as a strategy could distinguish $s_0$ and $s_2$ and assign different probabilities to $p$.
\end{example}
The pMC in the example is problematic as the parameters occur at the outgoing transitions of states in different combinations.
We restrict ourselves to an important subclass\footnote{All pMC benchmarks from the \tool{PARAM} webpage~\cite{param_website} are simple pMCs.} of pMCs which we call \emph{simple pMCs}.
A pMC is simple if for all states $s, s'$, $P(s, s') \in \Q \cup \{ p, 1-p \mid p \in \Var \}$.
Consequently, we can map states to parameters, and use this map to define the observations.
Then, the transformation from a POMDP to a pMC is the reverse of the transformation from Def.~\ref{def:pomdp_to_pmc_direct}.
In the remainder, we detail this correspondence.
The correspondence also establishes a construction to compute $k$-FSCs via parameter synthesis on \emph{simple} pMCs.
Current tool-support (cf.\ Sect.~\ref{sec:evaluation}) for simple pMCs is more mature than for the more general pMCs obtained via Def.~\ref{def:pomdp_to_pmc_direct}.

Let \textsl{simple-$\exists$INST} be the restriction of \textsl{$\exists$INST} to simple pMCs.
Similarly, let \textsl{simple-$\exists1$-FSC} be a variant of \textsl{$\exists$1-FSC} that only considers \emph{simple POMDPs}.

\begin{definition}[Binary/Simple POMDP]
	\label{def:simplepomdp}
	A POMDP is \emph{binary}, if $|\ActS(s)| \leq 2$ for all $s \in \states$.
	A binary POMDP is \emph{simple}, if for all $s \in \states$
	\[
	|\ActS(s)| = 2 \implies \forall a \in \ActS(s)~\exists s' \in \states: P(s, a, s') = 1.
	\]
\end{definition}
We establish the following relation between the POMDP and pMC synthesis problems, which asserts that the problems are equivalently hard.
\begin{theorem}
	$L_1 \polred L_2$ holds for any $L_1, L_2 \in \{$ \textsl{$\exists$k-FSC}, \textsl{$\exists$1-FSC}, \textsl{simple-$\exists$1-FSC}, \textsl{simple-$\exists$INST} $\}$.
\end{theorem}
The proof is a direct consequence of the Lemmas~\ref{lem:simplefsc-simplepmc}-\ref{lem:kfsc-1fsc} below, as well as the facts that every $1$-FSC is a $k$-FSC, and every simple POMDP is a POMDP.

The induced pMC $\dtmc_{\pomdp, 1}$ of a simple POMDP $\pomdp$ is also simple.
Consequently, Sect.~\ref{sec:verification} yields:
\begin{lemma}\label{lem:simplefsc-simplepmc}
	\textsl{simple-$\exists$1-FSC} $\polred$ \textsl{simple-$\exists$INST}.
\end{lemma}

\subsection{From Simple pMCs to POMDPs}
\begin{theorem}
	Every simple pMC $\dtmc$ with $n$ states and $m$ parameters is isomorphic to $\dtmc_{\pomdp,1}$ for some simple POMDP $\pomdp$ with $n$ states and $m$ observations.
\end{theorem}
We refrain from a formal proof: The construction is the reverse of Def.~\ref{def:pomdp_to_pmc_direct}, with
observations  $\{ \obs_p~|~p\in\VarMk[\dtmc]\}$.
In a simple pMC, the outgoing transitions are either all parameter free, or of the form $p, 1{-}p$.
The parameter-free case is transformed into a POMDP state with a single action (and any observation).
The parametric case is transformed into a state with two actions with Dirac-distributions attached. As observation we use $\obs_p$.
\begin{lemma}\label{lem:simplepmc-simplefsc}
\textsl{simple-$\exists$INST} $\polred$ \textsl{simple-$\exists$1-FSC}.
\end{lemma}
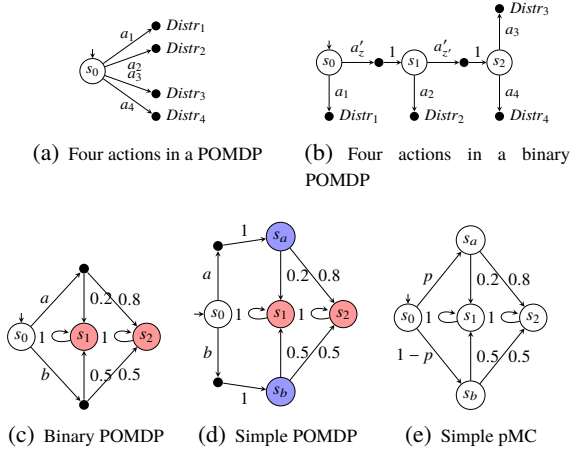
\begin{figure}[t]
	\centering
	\hspace*{\fill}
		\subfigure[\scriptsize Four actions in a POMDP]{
		\label{fig:BinPOMDP:orig}
		\scalebox{0.6}{
			\hspace*{1.5cm}
			\begin{tikzpicture}
	\node[state] (s0) {$s_0$};
	\node[circle, inner sep=2pt, fill=black, right=1cm of s0, yshift=1cm, label={right:$\Distr_1$}] (a0) {};
	\node[circle, inner sep=2pt, fill=black, right=1cm of s0, yshift=0.5cm, label={right:$\Distr_2$}] (a1) {};
	\node[circle, inner sep=2pt, fill=black, right=1cm of s0, yshift=-0.5cm, label={right:$\Distr_3$}] (a2) {};
	\node[circle, inner sep=2pt, fill=black, right=1cm of s0, yshift=-1cm, label={right:$\Distr_4$}] (a3) {};
	\draw ($(s0.north) + (0,0.2)$) edge[->] (s0);
	\draw[->] (s0) -- node[above] {$\act_1$} (a0);
	\draw[->] (s0) -- node[below,pos=0.65] {$\act_2$} (a1);
	\draw[->] (s0) -- node[above,pos=0.65] {$\act_3$} (a2);
	\draw[->] (s0) -- node[below] {$\act_4$} (a3);
\end{tikzpicture}
			\hspace*{1.5cm}
		}
	}
	\subfigure[\scriptsize Four actions in a binary POMDP]{
		\label{fig:BinPOMDP:linear}
		\scalebox{0.6}{
			\begin{tikzpicture}
	\node[state] (s0) {$s_0$};
	\node[state, right=1.3cm of s0] (s1) {$s_1$};
	\node[state, right=1.3cm of s1] (s2) {$s_2$};
	\node[circle, inner sep=2pt, fill=black, below=.8cm of s0,  label={right:$\Distr_1$}] (a0) {};
	\node[circle, inner sep=2pt, fill=black, right=.7cm of s0] (a1) {};
	\node[circle, inner sep=2pt, fill=black, below=.8cm of s1, label={right:$\Distr_2$}] (a2) {};
	\node[circle, inner sep=2pt, fill=black, right=.7cm of s1] (a3) {};
	\node[circle, inner sep=2pt, fill=black, above=.8cm of s2, label={right:$\Distr_3$}] (a4) {};
	\node[circle, inner sep=2pt, fill=black, below=.8cm of s2, label={right:$\Distr_4$}] (a5) {};
	
	\draw ($(s0.north) + (0,0.2)$) edge[->] (s0);
	\draw[->] (s0) -- node[right] {$\act_1$} (a0);
	\draw[->] (s0) -- node[above] {$\act'_\obs$} (a1);
	\draw[->] (s1) -- node[right] {$\act_2$} (a2);
	\draw[->] (s1) -- node[above] {$\act'_{\obs'}$} (a3);
	\draw[->] (s2) -- node[right] {$\act_3$} (a4);
	\draw[->] (s2) -- node[right] {$\act_4$} (a5);
	\draw[->] (a1) -- node[above] {$1$} (s1);
	\draw[->] (a3) -- node[above] {$1$} (s2);
\end{tikzpicture}
		}
	}
	\hspace*{\fill}
	\\
	\hspace*{\fill}
	\subfigure[\scriptsize Binary POMDP]{
		\label{fig:BinPOMDPtoSimplePMC:binaryPOMDP}
		\scalebox{0.65}{
			\begin{tikzpicture}
\node[state] (s0) {$s_0$};
\node[state, fill=red!40, right=0.7cm of s0] (s1) {$s_1$};
\node[state, fill=red!40, right=0.7cm of s1] (s2) {$s_2$};
\node[circle, inner sep=2pt, fill=black, above=1cm of s1] (a1) {};
\node[circle, inner sep=2pt, fill=black, below=1cm of s1] (a2) {};

\draw ($(s0.north) + (0,0.2)$) edge[->] (s0);
\draw[->] (s0) -- node[left] {$a$} (a1);
\draw[->] (s0) -- node[left] {$b$} (a2);
\draw[->] (a1) -- node[right] {$0.2$} (s1);
\draw[->] (a1) -- node[right] {$0.8$} (s2);
\draw[->] (a2) -- node[right] {$0.5$} (s1);
\draw[->] (a2) -- node[right] {$0.5$} (s2);
\path[->] (s1) edge  [loop left] node {$1$} ();
\path[->] (s2) edge  [loop left] node {$1$} ();
\end{tikzpicture}
		}
	}
	\subfigure[\scriptsize Simple POMDP]{
		\label{fig:BinPOMDPtoSimplePMC:simplePOMDP}
		\scalebox{0.65}{
			\begin{tikzpicture}
\node[state] (s0) {$s_0$};
\node[state, fill=red!40, right=0.7cm of s0] (s1) {$s_1$};
\node[state, fill=red!40, right=0.7cm of s1] (s2) {$s_2$};
\node[circle, inner sep=2pt, fill=black, above=1cm of s0] (a1) {};
\node[circle, inner sep=2pt, fill=black, below=1cm of s0] (a2) {};
\node[state, fill=blue!40, above=of s1] (as1) {$s_a$};
\node[state, fill=blue!40, below=of s1] (as2) {$s_b$};

\draw ($(s0.west) - (0.2,0)$) edge[->] (s0);
\draw[->] (s0) -- node[left] {$a$} (a1);
\draw[->] (s0) -- node[left] {$b$} (a2);
\draw[->] (a1) -- node[above] {$1$} (as1);
\draw[->] (a2) -- node[below] {$1$} (as2);
\draw[->] (as1) -- node[right] {$0.2$} (s1);
\draw[->] (as1) -- node[right] {$0.8$} (s2);
\draw[->] (as2) -- node[right] {$0.5$} (s1);
\draw[->] (as2) -- node[right] {$0.5$} (s2);
\path[->] (s1) edge  [loop left] node {$1$} ();
\path[->] (s2) edge  [loop left] node {$1$} ();
\end{tikzpicture}
		}
	}
	\subfigure[\scriptsize Simple pMC]{
		\label{fig:BinPOMDPtoSimplePMC:simplePMC}
		\scalebox{0.65}{
			\begin{tikzpicture}
\node[state] (s0) {$s_0$};
\node[state, right=0.7cm of s0] (s1) {$s_1$};
\node[state, right=0.7cm of s1] (s2) {$s_2$};
\node[state, above=of s1] (as1) {$s_a$};
\node[state, below=of s1] (as2) {$s_b$};

\draw ($(s0.north) + (0,0.2)$) edge[->] (s0);
\draw[->] (s0) -- node[left] {$p$} (as1);
\draw[->] (s0) -- node[left] {$1-p$} (as2);
\draw[->] (as1) -- node[right] {$0.2$} (s1);
\draw[->] (as1) -- node[right] {$0.8$} (s2);
\draw[->] (as2) -- node[right] {$0.5$} (s1);
\draw[->] (as2) -- node[right] {$0.5$} (s2);
\path[->] (s1) edge  [loop left] node {$1$} ();
\path[->] (s2) edge  [loop left] node {$1$} ();
\end{tikzpicture}
		}
	}
	\hspace*{\fill}

	\hspace*{\stretch{1}}
	\caption{POMDP $\leftrightarrow$ simple pMC}
	\label{fig:BinPOMDPtoSimplePMC}
\end{figure}
\subsection{Making POMDPs simple}
We present a reduction from \textsl{$\exists$1-FSC} to \textsl{simple-$\exists 1$-FSC} by translating a (possibly not simple) POMDP into a \emph{binary} POMDP and subsequently into a \emph{simple} POMDP.
Examples are given in Fig.~\ref{fig:BinPOMDPtoSimplePMC}(a--e).
We emphasise that our construction only preserves the expressiveness of $1$-FSCs.
\iftoggle{TR}{

There are several ways to transform a POMDP into a binary POMDP.
We illustrate one in Fig.~\ref{fig:BinPOMDPtoSimplePMC}(a--b).
The idea is to split actions at a state with more than two actions into two sets, which are then handled by fresh states with fresh observations.
The transformation iteratively reduces the number of actions until every state has
at most two outgoing actions.
To ensure a one-to-one correspondence between 1-FSCs of the original POMDP and the transformed POMDP, all states with the same observation should be handled the same way.
In particular, the same observations should be used for the introduced auxiliary states.

The transformation from binary POMDP to simple POMDP is illustrated by Fig.~\ref{fig:BinPOMDPtoSimplePMC}(c--d).
After each state with a choice of two actions, auxiliary states are introduced, such that the outcome of the action becomes deterministic and the probabilistic choice is delayed to the auxiliary state.
This construction is similar to the conversion of Segala's probabilistic automata into Hansson's alternating model~\cite{SegalaT05}.
Fig.~\ref{fig:BinPOMDPtoSimplePMC:simplePMC} shows the induced simple pMC.
}{Details are given in \cite{TR}.}

\begin{lemma}\label{lem:1fsctosimple1fsc}
    \textsl{$\exists$1-FSC} $\polred$ \textsl{simple-$\exists$1-FSC}.
\end{lemma}

\subsection{From \boldmath $k$-FSCs to $1$-FSCs}
For a POMDP $\pomdp$ and memory bound $k{>}1$, we construct a POMDP $\pomdp_k$ such that $\pomdp$ satisfies a specification $\varphi$ under some $k$-FSC iff $\pomdp_k$ satisfies $\varphi$ under some $1$-FSC.

\begin{definition}[$k$-Unfolding]
	\label{def:unfolding}
	Let $\PomdpInit$ be a POMDP with $\MdpInit$, and $k > 1$.
	The $k$-\emph{unfolding} of $\pomdp$ is the POMDP $\PomdpInit[_k]$ with $\mdp_k = (\states_k, \sinit{_{, k}}, \Act_k, \probmdp{_k})$ defined by:
	\begin{compactitem}
		\item $\states_k = \states \times \{0, \hdots k{-}1\}$, 
		\item $\sinit{_{, k}} = \langle \sinit, 0 \rangle$,
		\item $\Act_k = \Act \times \{0, \hdots, k{-}1\}$
		\item $\probmdp{_k}\bigl(\langle s, n \rangle,\langle \act, \bar{n} \rangle, \langle s', n' \rangle\bigr) {=} \begin{cases}
		\probmdp(s, \act, s') & n' = \bar{n}, \\
		0 & \text{else. }
		\end{cases}$
	\end{compactitem}
	and $\ObsSym_k = \ObsSym \times \{0, \hdots, k{-}1\}$, $\ObsFun_k\bigl(\langle s, n \rangle\bigr) = \bigl\langle \ObsFun(s), n \bigr\rangle$.
\end{definition}

Intuitively, $\pomdp_k$ stores the current memory node into its state space.
At state $\langle s, n \rangle$ of $\pomdp_k$, a $1$-FSC can not only choose between the available actions $\ActS(s)$ in $\pomdp$ but also between different successor memory nodes.

Fig.~\ref{fig:unfolding} shows this process for $k = 2$. All states of the POMDP are copied once.
Different observations allow to determine in which copy of a state---and therefore, which memory cell---we currently are.
Additionally, all actions are duplicated to model the option for a strategy to switch the memory cell.

\begin{figure}[t]
	\centering
	\hfill
	\subfigure[POMDP]{\scalebox{0.65}{
			\begin{tikzpicture}
			\node[state, fill=white!40] (s0) {$s_0$};
			\node[state, right=of s0, fill=white!40] (s1)  {$s_1$};
			\node[state, right=of s1, fill=red!40] (s2)  {$s_2$};
			\draw ($(s0.north) + (0,0.3)$) edge[->] (s0);
			\draw[->] (s0) -- node[above] {} (s1);
			\draw[->] (s1) -- node[above] {} (s2);
			\draw[->] (s2) edge[loop above] node[above] {} (s2);
			\draw[->] (s2) edge[bend right] node[above] {} (s0);
			
			\node[below=1cm of s0] {};
			
			\end{tikzpicture}
	}}\hfill
	\subfigure[2-Unfolding]{\scalebox{0.65}{
			\begin{tikzpicture}
			\node[state, fill=white!40] (s0) {$s_0$};
			\node[state, right=of s0, fill=white!40] (s1)  {$s_1$};
			\node[state, right=of s1, fill=red!40] (s2)  {$s_2$};
			\draw[->] (s0) -- node[above] {} (s1);
			\draw[->] (s1) -- node[above] {} (s2);
			\draw[->] (s2) edge[loop above] node[above] {} (s2);
			\draw[->] (s2) edge[bend right] node[above] {} (s0);
			\node[state, below=of s0, fill=blue!40] (t0) {$s_0$};
			\node[state, right=of t0, fill=blue!40] (t1)  {$s_1$};
			\node[state, right=of t1, fill=black!20] (t2)  {$s_2$};
			\draw ($(s0.north) + (0,0.3)$) edge[->] (s0);
			\draw[->] (t0) -- node[above] {} (s1);
			\draw[->] (t1) -- node[above] {} (s2);
			\draw[->] (t2) edge[bend left] node[above] {} (s2);
			\draw[->] (t2) edge[] node[above] {} (s0);
			\draw[->] (t0) -- node[above] {} (t1);
			\draw[->] (t1) -- node[above] {} (t2);
			\draw[->] (t2) edge[loop below] node[below] {} (t2);
			\draw[->] (t2) edge[bend left] node[above] {} (t0);
			\draw[->] (s0) -- node[above] {} (t1);
			\draw[->] (s1) -- node[above] {} (t2);
			\draw[->] (s2) edge[bend left] node[below] {} (t2);
			\draw[->] (s2) edge[] node[above] {} (t0);
			\end{tikzpicture}
	}}
\hspace{\fill}
	\caption{Unfolding a POMDP for two memory nodes}
	\label{fig:unfolding}
\end{figure}
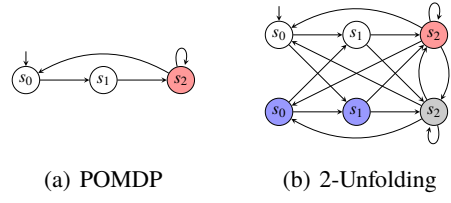

The induced pMC $\dtmc_{\pomdp_k,1}$ of the $k$-unfolding of $\pomdp$ has the same topology as the induced pMC $\dtmc_{\pomdp, k}$ of $\pomdp$ with memory bound $k$.
In fact, both pMCs have the same instantiations.
\begin{proposition}\label{prop:unfolding_instantiations_are_equal}
	For POMDP $\pomdp$ and memory bound $k$:
	\[
	\{\dtmc_{\pomdp_k,1}[u]~|~ u \text{ well-def.} \}
	=
	\{\dtmc_{\pomdp, k}[u]~|~ u \text{ well-def.} \}.
	\]
\end{proposition}

The intuition is that in both pMCs the parameter instantiations reflect arbitrary probability distributions over the same set of successor states.
In the transition probability function of the induced pMC $\dtmc_{\pomdp, k}$ of $\pomdp$ we can also substitute the multiplications of parameters $p_\act^{\obs,n}$ and $q_{\act, n'}^{\obs,n}$ by single parameters.
Consider the POMDP from Fig.~\ref{fig:pomdpfragment} and its induced pMC for memory bound $k=2$.
Tab.~\ref{tab:fsc_example_tab_subst} enumerates the outgoing transitions of the pMC state $\langle s_1,n_1 \rangle$ (cf.\ Fig.~\ref{fig:induced_pmc_example}).
We observe that the polynomials of the form $p\cdot q_i$ and
$p\cdot (1-q_i)$ for $i\in\{0,1\}$ are independent from each other.
We substitute them with single variables in the \emph{substituted} column.
The obtained pMC is called the \emph{substituted induced pMC}.
\begin{definition}[Substituted Induced pMC]
	\label{def:pomdp_to_pmc_substitute}
	Reconsider Def.~\ref{def:pomdp_to_pmc_direct}.
	We define $\pDtmcInitMkSubs$ by modifying $\VarMk$ and $H$ as follows:
	\begin{compactitem}
		\item $\VarMkSubs = \bigl\{r^{\obs,n}_{\act,n'} ~\big|~  \obs \in \ObsSym,  n,n' \in [k{-}1], \act \in \ActS(\obs)
		\text{ with } n' \neq k-1 \vee \act \neq \remaining(z) \bigr\}$,
		\item $H^{\mathrm{subs}}\bigl(\langle s, n \rangle, \langle s', n' \rangle, \act\bigr)= {}$
		\[
		\probmdp(s,\act,s') \cdot \left\{
		\begin{array}{lr}
		r^{z,n}_{\act,n'},\quad \text{if } \act \neq \remaining(z) \lor n' \neq k{-}1 \\
		1-\!\!\!\!\!\sum\limits_{b \neq \act \vee \bar{n}\neq k-1}r^{z,n}_{b,\bar{n}}, \\
		\phantom{r^{z,n}_{\act,n'},}\quad \text{if } \act = \remaining(z) \land n' = k{-}1
		\end{array}\right\}
		\]
		with $z = \ObsFun(s)$, and
		\item $P_{\pomdp,k}^{\mathrm{subs}}(s,s') = \sum_{\act \in \ActS(s)} H^{\mathrm{subs}}(s,s',\act)$ for all $s,s'\in\states_{\pomdp,k}$.
	\end{compactitem}
\end{definition}
The following proposition is a direct consequence:

\begin{proposition}
For POMDP $\pomdp$ and memory bound $k$:
	\[
\bigl\{ \MC_{\pomdp,k}[u]~\big|~ u \text{ well-def.} \bigr\} = \bigl\{ \MC^\text{subs}_{\pomdp,k}[u']~\big|~ u' \text{ well-def.} \bigr\}\,.
\]
\end{proposition}

\begin{table}[bth]
	\centering
	\caption{Substitution of polynomials in induced pMC}
        \label{tab:fsc_example_tab_subst}
	\scalebox{0.75}{
\begin{tabular}[b]{@{}|l|l|l|l|l|}
	\hline
	$\Act$                        & $\probmdp$                 & Node   & induced pMC & substituted \\ \hline\hline
	\multirow{4}{*}{$\act_1\colon p$}  & \multirow{2}{*}{0.6} & $n_1\colon q_1$ &  $0.6 \cdot p \cdot q_1$& 	$0.6 \cdot p_1$ 	\\ \cline{3-5}
	&                     & $n_2\colon 1-q_1$ 				&   $0.6 \cdot p \cdot (1{-}q_1)$    	&	$0.6 \cdot p_2$\\ \cline{2-5}
	& \multirow{2}{*}{0.4} & $n_1\colon q_1$ 				&  $0.4 \cdot p \cdot q_1$     		& 	$0.4 \cdot p_1$\\ \cline{3-5}
	&                      & $n_2\colon 1-q_1$ 				&  $0.4 \cdot p \cdot (1{-}q_1)$ 		&   $0.4 \cdot p_2$\\ \cline{1-5}
	\multirow{4}{*}{$\act_2\colon 1-p$} & \multirow{2}{*}{0.7} & $n_1\colon q_2$ 			&  $0.7 \cdot (1{-}p) \cdot q_2$  		&   $0.7 \cdot p_3$\\ \cline{3-5}
	&                      & $n_2\colon 1-q_2$ 				& $0.7 \cdot (1{-}p) \cdot (1{-}q_2)$	&   $0.7 \cdot (1{-}\sum_{i=1}^3p_i)$\\ \cline{2-5}
	& \multirow{2}{*}{0.3} & $n_1\colon q_2$ 				&  $0.3 \cdot (1{-}p) \cdot q_2$      	&	$0.3 \cdot p_3$\\ \cline{3-5}
	&                      & $n_2\colon 1-q_2$ 				&  $0.3 \cdot (1{-}p) \cdot (1{-}q_2)$ &   $0.3 \cdot (1{-}\sum_{i=1}^3p_i)$\\ \hline
\end{tabular}
}
\end{table}

Proposition~\ref{prop:unfolding_instantiations_are_equal} and Thm.~\ref{theorem:great_theorem_of_correspondence}
imply that induced MCs of $\pomdp$ under $k$-FSCs coincide with induced MCs of $\pomdp_k$ under $1$-FSCs:
$
\{ \pomdp^{\osched_\fsc}~|~\fsc \in \FSCs[k]{\pomdp} \}
=
\{ \pomdp_k^{\osched_\fsc}~|~\fsc \in \FSCs[1]{\pomdp} \}
$.

\begin{lemma}
    \label{lem:kfsc-1fsc}
    \textsl{$\exists$k-FSC} $\polred$ \textsl{$\exists$1-FSC}.
\end{lemma}

\section{Strategy Restrictions}
\label{sec:equivalence:restrictions}
Two simplifying restrictions on the parameters are usually made in parameter synthesis for pMCs:
\begin{compactitem}
  \item Each transition is assigned a strictly positive probability (\emph{graph-preserving}).
  \item Each transition is assigned at least probability $\varepsilon > 0$ (\emph{$\varepsilon$-preserving}).
\end{compactitem}
For simple pMCs, the restrictions correspond to parameters instantiations over $(0,1)$ or $[\varepsilon,1-\varepsilon]$, respectively.

Accordingly, we define restrictions to POMDP strategies that correspond to such restricted parameter instantiations.
\begin{definition}[Non-zero Strategies]
  A strategy $\sched$ is \emph{non-zero} if $\sched(\pi)(\act) > 0$ for all $\pi \in \pathsfin^{\mdp}, \act \in \ActS(\last{\pi})$, and \emph{min-$\varepsilon$}
  if $\sched(\pi)(\act) \geq \varepsilon > 0$.
\end{definition}
Non-zero strategies enforce $\supp(\sched(s)) = \ActS(s)$.
Example~\ref{ex:nonzero} shows the impact on reachability probabilities.
\begin{example}\label{ex:nonzero}
	The MDP $\mdp$ in Fig.~\ref{fig:nonZeroStrategy} has a choice between actions $\act_1$ and $\act_2$ at state $s_0$.
	If action $\act_1$ is chosen with probability zero, the probability to reach $s_1$ from $s_0$ becomes zero, and the corresponding parameter instantiation is not graph-preserving.
	Contrarily, if $\act_1$ is chosen with any positive probability, as would be enforced by a non-zero strategy, the probability to reach $s_1$ is one.
\end{example}
\begin{figure}[t]
\centering
\begin{tikzpicture}[font=\scriptsize]
\node[state, fill=white!40] (s0) {$s_0$};
\node[state, right=2cm of s0, fill=white!40] (s1)  {$s_1$};

\node[circle, inner sep=2pt, fill=black, left=0.7cm of s1] (a1) {};
\node[circle, inner sep=2pt, fill=black, left=0.7cm of s0] (a2) {};

\draw[-] (s0) -- node[above] {$\act_1$} (a1);
\draw[->] (a1) -- (s1);

\draw[-] (s0) edge[bend right=13] node[above] {$\act_2$} (a2);
\draw[->] (a2) edge[bend right=13] node[right] {} (s0);

\draw ($(s0.north) - (0,-0.2)$) edge[->] (s0);

\end{tikzpicture}
\caption{MDP $\pomdp$}
\label{fig:nonZeroStrategy}
\end{figure}

\begin{proposition}
  Let $\pomdp$ be a POMDP.
  An instantiation $u$ on $\dtmc_{\pomdp,1}$ is graph-preserving ($\varepsilon$-preserving), iff $\sched_{\fsc_u}$ is non-zero (min-$\varepsilon$).
\end{proposition}

Still, for the considered specifications, we can, w.\,l.\,o.\,g., restrict ourselves to FSCs that induce non-zero strategies.
\begin{theorem}
  Let $\pomdp$ be a POMDP, $k$ a memory bound and $\varphi=\p_{> \lambda} (\neg B\,\Until\,G)$.
  Either $\forall \fsc\in k\text{-FSC}: \pomdp^{\osched_\fsc} \not\models \varphi$ or
  $\exists \fsc'\in k\text{-FSC}: \pomdp^{\osched_{\fsc'}} \models \varphi$ with $\osched_{\fsc'}$ non-zero.
\end{theorem}

The theorem is a consequence of the following corresponding statement for pMCs.

\begin{lemma}
	\label{lem:graphpreserving}
	For pMC $\dtmc$ and $\varphi=\p_{> \lambda} (\neg B\,\Until\,G)$, either $\dtmc[u] \not\models \varphi$ for all well-defined instantiations $u$, or $\dtmc[u] \models \varphi$ for some graph-preserving instantiation.
\end{lemma}

\begin{proof}
Assume pMC $\pdtmc$ and $\varphi=\p_{> \lambda} (\neg B\,\Until\,G)$.
We have to show that either $\pdtmc[u] \not\models \varphi$ for all well-defined instantiations or $\pdtmc[u] \models \varphi$ for some graph-preserving instantiation.

Let $f$ be a function that maps a well-defined instantiation $u$ to the probability $\pr^{\dtmc[u]}(\neg B\,\Until\,G)$.
The essential idea is that the only reason for a discontinuity of $f$ is a change in the set $S_{=0}$---states in the pMC from which the probability to reach
the target is zero.
The number of states in $S_{=0}$ is the smallest under a graph preserving assignment.
A discontinuity of $f$ thus implies a reduced reachability probability (there are more states in $S_{=0}$).
\end{proof}

As a consequence, if we have to construct an instantiation which reaches
a goal with probability $>\kappa$, we can look for such an instantiation among the graph-preserving ones.
Moreover, the lemma also implies that the set of states $\states_{=0}$ can be precomputed.

\section{Alternative FSCs}\label{sec:alternative_fscs}

In the literature, several formalisms for FSCs occur.
In particular, ~\cite{ChatterjeeCGK16,meuleau1999solving,BK08} do not agree upon a common model.
We discuss the applicability of our results with respect to the different variants of FSCs.

\paragraph{Ignoring the Taken Action for Updates.}
In~\cite{meuleau1999solving,BK08}, the memory update is of the form
$\nodeTransition'\colon N \times Z \rightarrow \Distr(N)$. The update is a restriction of the
FSCs in this paper, represented by the constraint $\nodeTransition(n,\obs,\act_1) = \nodeTransition(n,\obs,\act_2)$.
The constraint yields dependencies between different actions, preventing the $k$-unfolding as in Def.~\ref{def:unfolding}.
\begin{example}
	Recall Ex.~\ref{ex:fsctopmc}, with the induced pMC for the POMDP fragment, as also given in Tab.~\ref{tab:fsc_example_tab}.
	Tab.~\ref{tab:fsc_tab_action_restricted} presents the induced pMC with the restriction in place. Notice that we have no parameter $q_2$ anymore.
	Based on Tab.~\ref{tab:fsc_example_tab} we could set  $p' = 0.5, q_1 = 0$, and the following transition probabilities for each target: $\{ \langle s_2, n_1 \rangle \mapsto 0, \langle s_2, n_2 \rangle \mapsto 0.3, \langle s_4, n_1 \rangle \mapsto 0.35  \}$.
	Based on Tab.~\ref{tab:fsc_tab_action_restricted}, this assignment is not possible.
\end{example}
We conclude from the example above that we get an additional parameter dependency.

\begin{definition}[Action-Restricted Induced pMC]
\label{def:pomdp_to_pmc_action_restricted}
	Reconsider Def.~\ref{def:pomdp_to_pmc_direct}.
        We define $\pDtmcInitMkRestr$ by modifying $\VarMk$ and $H^\mathrm{restr}$ as follows:
	\begin{itemize}
	\item $\VarMkRestr= \{ p^{\obs,n}_\act ~|~ \obs \in \ObsSym, n \in [k{-}1], \act \in \ActS(\obs), \act \neq \remaining(\obs) \}$ \\ $\cup \{ q^{\obs,n}_{n'} ~|~ n,n' \in [k{-}1], n' \neq k-1, \obs \in \ObsSym\}$
	\item $H^{\mathrm{restr}}\bigl(\langle s, n \rangle, \langle s', n' \rangle, \act\bigr)= {}$
        \begin{align*}
        \probmdp(s,\act,s') & \cdot \left\{\begin{array}{lr}
        p^{z,n}_\act, & \text{if } \act {\neq} \remaining\bigl(z\bigr)\\
        1-\sum\limits_{\act'\neq \act}p^{z,n}_{\act'}, & \text{if } \act {=} \remaining\bigl(z\bigr)
        \end{array}\right\} \\
       &\cdot \left\{\begin{array}{lr}
        q^{z,n}_{n'}, & \text{if } n'{\neq}k{-}1\\
        1-\sum\limits_{\bar{n}\neq n'}q^{z,n}_{\bar{n}}, & \text{if } n'{=}k{-}1
        \end{array}\right\}
    \end{align*}        with $z = \ObsFun(s)$
     \item $P_{\pomdp,k}^{\mathrm{restr}}(s,s') = \sum_{\act \in \ActS(s)} H^{\mathrm{next}}(s,s',\act)$ for all $s,s'\in S_{\pomdp,k}$.
    \end{itemize}
	The obtained pMC is then called the action-restricted induced pMC.
\end{definition}

For these pMCs, we can no longer perform the substitution as proposed in Def.~\ref{def:pomdp_to_pmc_substitute}.
As a consequence this restriction breaks the proposed unfolding.
\begin{table}[t]
\centering
\caption{Alternative induced pMCs}
\subfigure[Action restricted]{
\label{tab:fsc_tab_action_restricted}
\scalebox{0.8}{
\begin{tabular}{|l|l|l|l|l|}
\hline
Obs                    & Act                        & P                 & Node   & Result \\ \hline\hline
\multirow{8}{*}{$\obs_1$} & \multirow{4}{*}{$\act_1:p'$}  & \multirow{2}{*}{0.6} & $n_1:q_1$ &  $0.6 \cdot p' \cdot q_1$      \\ \cline{4-5}
                       &                            &                      & $n_2:1-q_1$ &   $0.6 \cdot p' \cdot (1{-}q_1)$     \\ \cline{3-5}
                       &                            & \multirow{2}{*}{0.4} & $n_1:q_1$ &  $0.4 \cdot p' \cdot q_1$      \\ \cline{4-5}
                       &                            &                      & $n_2:1{-}q_1$ &  $0.4 \cdot p' \cdot (1{-}q_1)$      \\ \cline{2-5}
                       & \multirow{4}{*}{$\act_2:1{-}p'$} & \multirow{2}{*}{0.7} & $n_1:\mathbf{\color{red}q_1}$ &  $0.7 \cdot (1{-}p') \cdot q_1$      \\ \cline{4-5}
                       &                            &                      & $n_2:\mathbf{\color{red}1{-}q_1}$ & $0.7 \cdot (1{-}p') \cdot (1{-}q_1)$       \\ \cline{3-5}
                       &                            & \multirow{2}{*}{0.3} & $n_1:\mathbf{\color{red}q_1}$ &  $0.3 \cdot (1{-}p') \cdot q_1$      \\ \cline{4-5}
                       &                            &                      & $n_2:\mathbf{\color{red}1{-}q_1}$ &  $0.3 \cdot (1{-}p') \cdot (1{-}q_1)$      \\ \hline
\end{tabular}
}
}
\subfigure[Next observation dependent]{
\scalebox{0.8}{
\begin{tabular}{|l|l|l|l|l|}
\hline
Obs                    & Act                        & P                 & Node   & Result \\ \hline\hline
\multirow{8}{*}{$\obs_1$} & \multirow{4}{*}{$\act_1:p'$}  & \multirow{2}{*}{0.6} & $n_1:\mathbf{\color{red}q_1}$ &  $0.6 \cdot p' \cdot q_1$      \\ \cline{4-5}
                       &                            &                      & $n_2:\mathbf{\color{red}1{-}q_1}$ &   $0.6 \cdot p' \cdot (1{-}q_1)$     \\ \cline{3-5}
                       &                            & \multirow{2}{*}{0.4} & $n_1:q_2$ 				&  $0.4 \cdot p' \cdot q_2$      \\ \cline{4-5}
                       &                            &                      & $n_2:1-q_2$ 				&  $0.4 \cdot p' \cdot (1{-}q_2)$      \\ \cline{2-5}
                       & \multirow{4}{*}{$\act_2:1{-}p'$} & \multirow{2}{*}{0.7} & $n_1:\mathbf{\color{red}q_1}$ &  $0.7 \cdot (1{-}p') \cdot q_1$      \\ \cline{4-5}
                       &                            &                      & $n_2:\mathbf{\color{red}1{-}q_1}$ & $0.7 \cdot (1{-}p') \cdot (1{-}q_1)$       \\ \cline{3-5}
                       &                            & \multirow{2}{*}{0.3} & $n_1:\mathbf{\color{red}q_1}$ &  $0.3 \cdot (1{-}p') \cdot q_1$      \\ \cline{4-5}
                       &                            &                      & $n_2:\mathbf{\color{red}1{-}q_1}$ &  $0.3 \cdot (1{-}p') \cdot (1{-}q_1)$      \\ \hline
\end{tabular}
}
}
\end{table}

\paragraph{Taking the Next Observation into Account.}
In this paper, the memory node update in FSCs depends on the observation at the state \emph{before} executing the action.
Instead, the update may also be based on the observation \emph{after} the update~\cite{meuleau1999solving}.
This notion introduces dependencies between actions from states with different observations that reach the same observation.
\begin{example}
  Recall Ex.~\ref{ex:fsctopmc}, with the induced pMC for the POMDP fragment, as also given in Tab.~\ref{tab:fsc_example_tab}.
  Tab.~\ref{tab:fsc_tab_action_restricted} presents the induced pMC with the restriction in place. Notice that the memory update probabilities now depend on the observation of the resulting state. In particular, the action probability depends on the current observation, and features dependencies between source states, while the memory update features dependencies between target states.
\end{example}
\begin{definition}[Next-observation induced pMC]
    \label{def:pomdp_to_pmc_next}
    Reconsider Def.~\ref{def:pomdp_to_pmc_direct}.
    We define $\pDtmcInitMkNext$ by:
    \begin{itemize}
    \item Extending the set of variables of $\VarMk$,  $\VarMkNext = $ \begin{align*} 
& \bigl\{ p^{\obs,n}_\act ~\big|~ \obs \in \ObsSym, n \in [k{-}1], \act \in \ActS(\obs), \act \neq \remaining(\obs) \bigr\}\\ & \cup  \bigl\{ q^{\obs,n}_{\act,n'} ~\big|~\obs \in \ObsSym,  n,n' \in [k{-}1], n' {\neq} k{-}1, \act \in \ActS(z) \bigr\}
 \end{align*}
    \item Using the second set of variables for the  memory update: \\ $H^{\mathrm{next}}\bigl(\langle s, n \rangle, \langle s', n' \rangle, \act\bigr)= {}$
        \begin{align*}
        \probmdp(s,\act,s') & \cdot
        \left\{\begin{array}{lr}
           p^{z,n}_\act, & \text{if } \act {\neq} \remaining\bigl(z\bigr)\\
           1-\sum\limits_{\act'\neq \act}p^{z,n}_{\act'}, & \text{if } \act {=} \remaining\bigl(z\bigr)
        \end{array}\right\} \\
        & \cdot
        \left\{\begin{array}{lr}
           q^{z',n}_{\act,n'}, & \text{if } n'{\neq}k{-}1\\
           1-\sum\limits_{\bar{n}\neq n'}q^{z',n}_{\act,\bar{n}}, & \text{if } n'{=}k{-}1
        \end{array}\right\}
    \end{align*}
    with $z = \ObsFun(s), z' = \ObsFun(s')$, and
    \item $\begin{aligned}[t]P_{\pomdp,k}^{\mathrm{next}}(s,s') = \sum_{\act \in \ActS(s)} H^{\mathrm{next}}(s,s',\act) \text { for all } s,s'\in S'\end{aligned}$.
    \end{itemize}
    The obtained pMC is then called the \emph{next-induced pMC}.
\end{definition}

Due to the dependencies, we cannot substitute monomials, and we cannot simply unfold the memory into the POMDP.

We observe that compared to taking the next observation into account, the defined FSC lags behind, and needs an additional step.
Thus, we can modify the POMDP to give the memory structure time to update, by introducing an intermediate state with an adequate observation after every action.
On the transformed POMDP, this alternative FSC behaves as the FSC considered in the rest of the paper.

\paragraph{Ignoring the Current Observation when Selecting the Action.}
In~\cite{ChatterjeeCGK16}, the action mapping is modeled as $\actionMap'\colon N \rightarrow \Distr(\Act)$,
which restricts our FSC to $\actionMap(n, \obs) = \actionMap(n, \obs')$.
This type of FSC is more general in the sense that it can assign memory usage more freely
than the rather uniform assignment used here. In particular, a model with one memory node is
now not memoryless anymore, but weaker (it has to select the same action distribution
regardless of the observation). It also contains some restrictions: In particular, every POMDP
state requires the same action set.
Therefore, this model is not compatible with our framework.

\section{Permissive Strategies}
\label{sec:solving:lifting}\label{sec:permissive}

\emph{Permissive strategies} are sets of strategies satisfying a specification for MDPs~\cite{DBLP:journals/corr/DragerFK0U15,junges-et-al-tacas-2016}.
One interesting application of permissive strategies is \emph{robustness} in the sense that we can assess if a slight change to a strategy will preserve the satisfaction of the specification.
For observation-based strategies, a permissive scheduler may be defined as a function $\pathsfin^{\mdp}\to 2^{\Distr(\Act)}$ mapping paths to (infinite) sets of distributions over actions.
Here, we restrict ourselves to sets of $k$-FSCs.
\begin{definition}
  A \emph{permissive $k$-FSC} for a POMDP $\pomdp$ and a specification $\varphi$ is a subset $\PFSC[k]{\pomdp}{\varphi} \subseteq \FSCs[k]{\pomdp}$
  such that $\pomdp^{\sched_\fsc} \models \varphi$ for each FSC $\fsc \in \PFSC[k]{\pomdp}{\varphi}$.
\end{definition}
Using the correspondence from Theorem~\ref{theorem:great_theorem_of_correspondence}, we transfer the notion of permissive schedulers to \emph{sets of parameter instantiations} that satisfy a given specification on the induced pMCs.
We aim at finding (preferably large) regions of well-defined instantiations that all satisfy the given specification.

In parameter synthesis, similar problems are typically tackled in two steps.
We first find a suitable candidate permissive $k$-FSC. Then, we prove that the candidate indeed is an permissive $k$-FSC.
The latter problem is a de-facto universal quantification problem, which can be rephrased such that the objective is to prove the absence of an instantiation which does not satisfy the specification.
Finding a suitable candidate FSC corresponds to finding a region, which is typically guided by sampling.
Thus, we propose to first find several $k$-FSCs that satisfy the specification, and take a region containing these FSCs and no FSCs that do not satisfy the specification as a candidate.
We leave dedicated techniques to finding permissive FSC candidates as future work.

\section{Empirical Evaluation}
\label{sec:evaluation}

We established the correspondence between the synthesis problems for POMDPs and pMCs.
Now, we discuss the available methods for pMC parameter synthesis, and how they may be exploited or adapted to synthesise FSCs.
We distinguish three key problems:
\begin{enumerate}[nosep,itemindent=0pt,leftmargin=0pt,itemindent=3em]
 \item
\emph{Find a correct-by-construction strategy for a POMDP and a specification.}
To construct such a strategy, one needs to find a parameter valuation for the pMC that provably satisfies the specification.
Most solution techniques focused on pMCs with a few parameters, rendering the problem at hand infeasible.
Recently, efficient approaches emerged that are either based on particle swarm optimisation (PSO)~\cite{DBLP:conf/tase/ChenHHKQ013} or on convex optimisation~\cite{amato2010optimizing,DBLP:conf/tacas/Cubuktepe0JKPPT17}, in particular using quadratically-constrained quadratic programming (QCQP)~\cite{cubuktepe-et-al-atva-2018}.
We employ PSO and QCQP for our evaluation.

\item \emph{Prove that no FSC exists for a POMDP and a specification.}
Proving the absence of an FSC with the given memory bound allows us to show $\varepsilon$-optimality of a previously synthesised strategy.
Two approaches exist: An approximative technique called \emph{parameter lifting}~\cite{QDJJK16} and a method based on SAT-modulo-theories (SMT) solving~\cite{dMB08}.

\item \emph{Provide a closed-form solution} for the underlying measure of a specification in form of a function over the induced parameters of an FSC.
The function may be used for further analysis, e.\,g.\ of the sensitivity of decisions or parameter values, respectively.
To compute this function, all of the parameter synthesis tools \tool{PARAM}~\cite{param_sttt}, \tool{PRISM}~\cite{KNP11}, \storm~\cite{DBLP:conf/cav/DehnertJK017}, and \prophesy~\cite{DJJ+15} employ a technique called \emph{state elimination}~\cite{Daw04}.
\end{enumerate}

\paragraph{Implementation and Setup.}
We extended the tool \storm~\cite{DBLP:conf/cav/DehnertJK017} to parse and store POMDPs, and implemented several transformation options to pMCs.
Most notably, \storm supports $k$-unfolding, the product with several restricted FSCs such as counters that can be incremented at will, and several types of transformation to (simple) pMCs.

We evaluated on a HP BL685C G7 with 48 2\,GHz cores, a 16\,GB memory limit, and 1800~seconds time limit.
The compared methods are single-threaded.
We took the POMDPs from PRISM-POMDP~\cite{NPZ17}, additional maze, load/unload examples from~\cite{meuleau1999solving}
and a slippery gridworld with traps inspired by~\cite{DBLP:books/daglib/0023820}.
Table~\ref{tab:instances} gives details.
The specifications (Tp.) either ask to minimise  expected costs from an initial state until reaching a specified target set ($\mathbb{E}$)
or maximising the probability of reaching from an initial state a target set without hitting a bad state before ($\mathbb{P}$).
We list the number of states, branches ($\sum |A(s)|$) and observations in each POMDP.
As a baseline, we provide the results and run time of the model-checking tool PRISM-POMDP, and the point-based solver SolvePOMDP~\cite{DBLP:conf/aaai/WalravenS17}, obtained with default settings.
Both tools compute optimal memory-unbounded strategies and are prototypes.
The last column contains the result on the underlying fully observable MDP.
The experiments contain minimal expected rewards, which are analysed by a straightforward extension of maximal reachability probabilities.
All pMCs computed are simple pMCs, as \prophesy typically benefits from the simpler structure. \prophesy has been invoked with the default set-up.

\begin{table}[tb]
\centering
\setlength\tabcolsep{2pt}
\caption{Benchmarks}
\label{tab:instances}
\scalebox{0.9}{
{\scriptsize
\begin{tabular}{|lll|rrr|rr|rr|r|}
\hline
\multicolumn{3}{|c|}{} & \multicolumn{3}{c|}{POMDP $\pomdp$} & \multicolumn{2}{c|}{PRISM-POMDP}  & \multicolumn{2}{c|}{SolvePOMDP} & MDP  \\
Id & Name       & Tp.	& States & Bran.  & Obs. & Result & Time & Result & Time &  Res\\\hline\hline
 1  & NRP (8)  & $\mathbb{P}$	& 125       &  161          &  41   &  $[.125,.24]$    &  20     &   & \TO  & 1.0 \\
 2  & Grid (4)  &$\mathbb{E}$	&    17    &  62            & 3    &  $[3.97,4.13]$    &  $1038$     & $4.13$ &  $0.4$ & 3.2\\
 3  & Netw (3,4,8) & $\mathbb{E}$   	&   2729            &  4937     &  361   &       &    \TO   &  & \TO  &  0.83 \\
 4  & Crypt (5)  &  $\mathbb{P}$	&   4885   &  11733    &  890     &        &    \MO   &      & \TO  & 1.0\\
 5  & Maze (2)   &   $\mathbb{E}$     &   16     & 58              & 8    &  $[5.11,5.23]$    &    $3.9 $  & $5.23$ & 16  & 4.0\\
 6  & Load (8)  & $\mathbb{E}$ 	&    16    &   28            & 5    &     $[10.5,10.5]$   &    $1356$   & $10.5$ & 7.6  & 10.5\\
 7  & Slippery (4) &  $\mathbb{P}$  	&   17     &  59          &  4   &       &  \TO      &  0.93 & 95  & 1.0\\
\hline
\end{tabular}
}}
\end{table}

\subsection{Finding strategies}
\label{sec:evaluation:finding}

We evaluate how quickly a strategy that satisfies the specification is synthesised.
We vary the threshold used in the specification, as well as the structure of the FSC.

\paragraph{Results.}
\begin{table}[tb]
\centering
\setlength\tabcolsep{1.9pt}
\caption{Synthesing strategies}
\label{tab:finding-strategies}
\scalebox{0.9}{
\scriptsize{
\begin{tabular}{|cc|llll|rr|rr|rr|}
\hline
Id$\!\!$                 & Ts                         & FSC/$k$ & States & Trans & Pars & \multicolumn{2}{c|}{T1}         & \multicolumn{2}{c|}{T2}         & \multicolumn{2}{c|}{T3} \\
                   &                            &     &           &       &     & pso & \multicolumn{1}{l|}{qcqp} & pso & \multicolumn{1}{l|}{qcqp} & pso       & qcqp       \\ \hline\hline
\multirow{4}{*}{1} & \multirow{4}{*}{.124/.11/.09} &  F/1  & 75   & 118 & 8 & $<$1 & $<$1  & $<$1 & $<$1 &    $<$1    & $<$1        \\
                   &                            &   F/2 &  205      &  420     &   47                        &   2  &        $<$1                   &  2   &     $<$1                     &       2    &  $<$1         \\
                   &                            &   F/4  &  921      &  1864     &  215                         &   9  &       2                    &  9   &   2                         &  10         &  2          \\
                   &                            &   F/8  &   3889     &  7824     &   911                        &   43  &    15                 &  42   &   14    & 42 & 14 \\\hline
\multirow{4}{*}{2} & \multirow{4}{*}{4.15/4.5/5.5} &  F/1  & 47   & 106 & 3 & -- & --  & -- & --  &  \Err      & $<$1        \\
                   &                            &   F/2 &  183      &  390     &   15                        &   7.4  &        11                   &   4  &     9                      &    2       &   $<$1         \\
                   &                            &   F/4  &  719      &  1486     &  63                         &  \TO   &       64                    &  39   &   91                         &  14         &  8          \\
                   &                            &   F/8  &   2845     &  5788     &   255                        & \TO    &     700                      & \TO     &   946    & 254 & 69 \\\hline
\multirow{4}{*}{3} & \multirow{4}{*}{9/10/15} &  F/1  & 3268   & 13094 & 276 & \TO & \TO  & \TO & 43 &  22      & 4       \\
                   &                            &   F/2 &  16004      &  46153     &   1783                        &   \TO  &        \TO                   &   \TO   &     877                      &  152         &   28         \\\cline{3-12}
                  & &   C/2 &  11270      &  36171     &   1168                        &   \TO  &        \TO                   &   \TO   &     358                      &  100         &   62         \\
                  & &   C/4 &  27183      &  82145     &   2940                        &  \TO   &        \MO                   &  \TO    &     \MO                      &  476         &   \MO         \\\hline
\multirow{2}{*}{4} & \multirow{2}{*}{.249/.2/.15} &  F/1  & 3366   & 6534 & 364 & 18 & 25  & 18 & 15 &   18     & 12       \\
                   &                            &   F/2 &  25713      &  51608     &   3907                        &   330  &        \MO                   &  350    &     \MO                      &  326         &   \MO         \\\hline
\multirow{6}{*}{5} & \multirow{6}{*}{5.2/15/25} &  F/1  &  30     & 64 			& 8 				& -- & --  & \TO  & \TO  & $<$1        &  \TO        \\
                   &                            &   F/2 &  137      &  294     		&   49              & \TO    & \TO   &  14   &     \TO                      &  2         &  \TO          \\
                   &                            &   F/4  &  587      &  1214     		&  219              & 93    & \TO    &  \TO   &   \TO                         &    26       &   \TO         \\
                   &                            &   F/8  &   2421     &  4924     	&   919             & \TO    & \TO     & 1034    &  \TO     & 115 & \TO \\\cline{3-12}
                    &                            &   C/2  &   99     &  212     	&   33             & \TO    & \TO     & 3.7    &   \TO    & $<$1   & \TO \\
                     &                            &   C/4  &   231     &  476     	&   81             & 7    & \TO     & 6    &  \TO     &  3 & \TO  \\\hline
\multirow{3}{*}{6} & \multirow{3}{*}{10.6/10.9/82.5} &  F/1  & 16   & 33 & 1 & -- & --  & -- & -- &   $<$1     & \TO       \\
                   &                            &   F/2 &  77      &  160     &   11                        &   9  &         \TO                  &  6    &     \TO                      &  $<$1         &  \TO          \\
                    &                            &   F/4 &  354      &  721     &   63                        &   20  &    \TO                       &  21    &            63               &  3         &   \TO         \\\hline
\multirow{4}{*}{7} & \multirow{4}{*}{.929/.928/.927} &  F/1  & 87   & 184 & 3 & \TO & \TO  & $<$1 & 1 & $<$1       & $<$1        \\
                   &                            &   F/2 &  285      &  592     &   15                        &    4 &        \TO                   &    4 &     20                      &           3 &   22         \\
                   &                            &   F/4  &  1017      &  2080     &  63                       &   76  &       767                    & 71    &   205                         &     67      &  187          \\
                   &                            &   F/8  &   3825     &  7744     &   255                     &   \TO  &     \TO                      &  \TO   &   \TO    &  \TO & \TO \\\hline
\end{tabular}}}
\end{table}

We summarise the obtained results in Tab.~\ref{tab:finding-strategies}.
For each instance (Id), we define three thresholds (\textsf{Ts}), ordered from challenging (i.\,e.\ close to the optimum) to less challenging.
For different types of FSCs (FSC, F=full, C=counter) and memory bounds ($k$), we obtain pMCs with the given number of states, transitions and parameters. Full-FSCs are fully connected, in counter-FSCs memory node $m$ is succeeded by either $m$ or $m+1$.
For each threshold ($\textsf{T1}$, $\textsf{T2}$, $\textsf{T3}$), we report the run time of the two methods PSO and QCQP, respectively.
\textsf{T1} is chosen to be nearly optimal for all benchmarks.
A dash indicates a combination of memory and threshold for which no FSC exists, according to the results in Sect.~\ref{sec:evaluation:proving}.
TO/MO denote violations of the time/memory limit, respectively.

\paragraph{Evaluation.}
Strategies for thresholds which are suboptimal (\textsf{T3}) are synthesised faster.
If the memory bound is increased, the number of parameters quickly grows and the performance of the methods degrades.
Additional experiments showed that the number of states has only a minor effect on the performance.
The simpler FSC topo\-logy for a counter alleviates the blow-up of the pMC and is successfully utilised to find strategies for larger instances.

Trivially, a $k$-FSC is also a valid $(k{+}i)$-FSC for some $i \in \mathbb{N}$. Yet, the larger number of parameters make searching for $(k{+}i)$-FSCs significantly more difficult.
We furthermore observe that the performance of PSO and QCQP is incomparable, and both methods have their merits.

Summarising, many of the POMDPs in the benchmarks allow good performance via FSCs with small memory.
\textbf{We find nearly-optimal, and small, FSCs for POMDP benchmarks with thousands of states within seconds.}

\subsection{Proving \boldmath $\varepsilon$-Optimality}
\label{sec:evaluation:proving}
We now focus on evaluating how fast pMC techniques prove the absence of a strategy satisfying the specification.
In particular, we consider proving that for a specific threshold, no strategy induces a better value.
Such a proof allows us to draw conclusions about the
\mbox{($\varepsilon$-)}
optimality of a strategy synthesised in Sect.~\ref{sec:evaluation:finding}.

\paragraph{Results.}
\begin{table}[]
\centering
\caption{$\varepsilon$-optimality and closed-form computation}
\subtable[Proving absence]{
\label{tab:absence}
\scalebox{0.9}{
\scriptsize{
\begin{tabular}{|llrr|}
\hline
 Id & FSC/$k$ & T & time       \\\hline\hline
  2 & F/1   & 5    & $<$1    \\\hline
  3 & F/1   & 5    & 8       \\
  3 & F/4   & 5    & 183     \\\hline
  4 & F/1   & 0.25 & 2$^{*}$ \\\hline
  5 & F/1   & 10   & 3       \\
  5 & F/2   & 5    & \TO     \\\hline
  6 & F/1   & 82   & $<$1    \\
  6 & F/8   & 10.5 & 1       \\\hline
  7 & F/1   & 0.94 & 5       \\\hline
\end{tabular}
}}
}\qquad
\subtable[Closed-form solution]{
\label{tab:closedform}
\quad\scalebox{0.9}{
\scriptsize{
\begin{tabular}{|llr|}
\hline
 Id & FSC/$k$ & time \\\hline\hline
  1 & F/1   & $<$1 \\
  1 & F/2   & 97   \\\hline
  2 & F/1   & 155  \\\hline
  3 & F/1   & 464  \\\hline
  4 & F/1   & $<$1 \\\hline
  5 & F/1   & 116  \\\hline
  6 & F/1   & $<$1 \\\hline
  7 & F/1   & \TO  \\\hline
\end{tabular}
}
}\quad}
\end{table}

Table~\ref{tab:absence} shows the run times to prove that for the POMDP in column Id, there exists no strategy of type FSC with $k$ memory that performs better than threshold $T$.
The row indicated by $^{*}$ was obtained with SMT.
All other results were obtained with parameter lifting.
\paragraph{Evaluation.}
The methods generally prove tight bounds for $k{=}1$. For $k{>}1$, the high number of parameters yields a mixed impression, the performance depends on the benchmark.
\textbf{We find proofs for non-trivial bounds up to \boldmath $k=8$, even if the pMC has hundreds of parameters.}

\subsection{Closed-form solutions}
\paragraph{Results.}
Table~\ref{tab:closedform} indicates run times to compute a closed-form solution, i.\,e.\ a rational function that maps $k$-FSCs to the induced probability.

\paragraph{Evaluation.}
Closed form computation is limited to small memory bounds.
The rational functions obtained vary wildly in their structure.
For (4), the result is a constant function, which is trivial to analyse, while for (3), we obtained rational functions with roughly one million terms, rendering further evaluation expensive.

\section{Conclusion}
\label{sec:conclusion}

This paper connects two active research areas, namely verification and synthesis
for POMDPs and parameter synthesis for Markov models. We see benefits for both areas.
On the one hand, the rich application area for POMDPs in, \eg\ robotics, yields
new challenging benchmarks for parameter synthesis and can drive the development of
more efficient methods. On the other hand, parameter synthesis tools and techniques
extend the state-of-the-art approaches for POMDP analysis. Future work will also concern
a thorough investigation of \emph{permissive schedulers}, that correspond to regions of
parameter instantiations, in concrete motion planning scenarios.

\bibliographystyle{splncs03}
\balance
\begin{small}
\bibliography{abbrev_short,literature}
\end{small}

\end{document}